\documentclass[conference]{IEEEtran}
\usepackage{url,caption}
\usepackage{amsmath,amssymb,amsfonts,amsthm}
\usepackage{tikz}
\usepackage{filecontents}
\usepackage{textcomp}
\usepackage{pifont,framed}
\usepackage[noend]{algpseudocode}
\usepackage{algorithm,algorithmicx}
\usepackage{xcolor}
\usepackage{color}
\usepackage{verbatim}
\usepackage[algo2e,ruled,vlined]{algorithm2e}
\usepackage{multirow,subfigure,graphicx}

\newtheorem{lemma}{\bf Lemma}
\newtheorem{axiom}{\bf Axiom}

\newtheorem{definition}{\bf Definition}
\newtheorem{theorem}{\bf Theorem}

\begin{document}

\title{Interactive Trimming against Evasive Online Data Manipulation Attacks: A Game-Theoretic Approach}

\author{
\IEEEauthorblockN{Yue Fu, Qingqing Ye, Rong Du, Haibo Hu}

\IEEEauthorblockA{The Hong Kong Polytechnic University}
Email: \{yuesandy.fu, roong.du\}@connect.polyu.hk, \{qqing.ye, haibo.hu\}@polyu.edu.hk}

\maketitle

\begin{abstract}
With the exponential growth of data and its crucial impact on our lives and decision-making, the integrity of data has become a significant concern. Malicious data poisoning attacks, where false values are injected into the data, can disrupt machine learning processes and lead to severe consequences. To mitigate these attacks, distance-based defenses, such as trimming, have been proposed, but they can be easily evaded by white-box attackers. The evasiveness and effectiveness of poisoning attack strategies are two sides of the same coin, making game theory a promising approach. However, existing game-theoretical models often overlook the complexities of online data poisoning attacks, where strategies must adapt to the dynamic process of data collection.

In this paper, we present an interactive game-theoretical model to defend online data manipulation attacks using the trimming strategy. Our model accommodates a complete strategy space, making it applicable to strong evasive and colluding adversaries. Leveraging the principle of least action and the Euler-Lagrange equation from theoretical physics, we derive an analytical model for the game-theoretic process. To demonstrate its practical usage, we present a case study in a privacy-preserving data collection system under local differential privacy where a non-deterministic utility function is adopted. Two strategies are devised from this analytical model, namely, Tit-for-tat and Elastic. We conduct extensive experiments on real-world datasets, which showcase the effectiveness and accuracy of these two strategies.
\end{abstract}

\section{Introduction}
\label{Introduction}
In the era of big data and AI, the sheer volume and ubiquity of data have profound impact on our daily life and the world at large. As such, the integrity of data stands as a cornerstone for high-quality data analysis and decision making. Unfortunately, data integrity is under perpetual threat --- malicious entities frequently engage in data manipulation, fabricating falsified values to skew outcomes in their favor.

The issue of data manipulation has been a focal point in the data management community, for example in the field of 
knowledge graph~\cite{banerjee2021stealthy, zhu2022binarizedattack}, federated recommendation systems~\cite{song2020poisonrec,rong2022fedrecattack,wu2022fedattack}, and countermeasures~\cite{yeh2023planning}. Data manipulation attacks also pose immediate threats to the training process of machine learning systems. Given these high stakes, it is imperative for data collectors to take safeguard measures to detect and neutralize data poisoning attacks, while retaining good quality of the rest (benign) data. 

To reduce the impact of data manipulation attacks, one approach is to sanitize the input dataset. A classic method is distance-based sanitization, also known as trimming, where the defender calculates the distance $d_i$ for each data point $i$ and removes any point with $d_i > \theta_d$, a threshold chosen by the defender~\cite{kloft2012security}. Popular distance-based defenses against data manipulation attacks include~\cite{koh2022stronger,laishram2016curie}, by optimizing a designated objective function. However, such strategies are static and neglect the evasive nature of adversaries, that is, they always manage to circumvent these defensive measures~\cite{ou2019mixed}. Therefore, an evasion-aware defense strategy must consider potential evasion strategies employed by these adversaries. Game theory is a common tool to find a dynamic balance between evasive attackers and defenders, known as Nash equilibrium. Recently, a few game-theoretical models~\cite{ou2019mixed,zhang2017game} have been proposed for {\bf static} data poisoning attacks, where data are collected in a single round. However, in many real-life data collection systems, data are frequently updated or streaming, and the collection process is continual or in multi-round. As such, a static defensive strategy is insufficient, as adversaries can adapt their strategies in each round. Due to the immense complexity of potential strategies a dynamic attacker might deploy, it has rarely been explored in the field of {\bf online data manipulation attacks}.

In this paper, our objective is to derive a feasible Stackelberg equilibrium within a \textbf{complete trimming strategy space} to defend against data manipulation attacks, specifically in the context of online data poisoning. Our game-theoretic model is anchored in the simplicity of the trimming strategy and is shown to achieve a genuine equilibrium within its complete strategy space. The findings are validated using real-world machine learning data across widely-used algorithms, including k-means, SVM, and SOM classification. We illustrate how the threshold and poison values are determined and elucidate the impact of each scheme on the system's final outcome. Additionally, through empirical studies, we illustrate that \textbf{attackers who behave irrationally and diverge from rational strategies will merely gain less utility from poison values}. The key contributions of our work can be encapsulated as follows:

$\cdot$ We propose an interactive game-theoretic model for online data poisoning attacks and defenses using the trimming mechanism. This model streamlines the formulation process, accommodates a complete strategy spectrum, and simplifies the derivation of Stackelberg equilibrium, even against evasive and colluding attackers with diversified poisoning strategies.

$\cdot$ We utilize the principle of least action and the Euler-Lagrange equation in theoretical physics to build an analytical model for the game-theoretic process. The model is in the form of the Lagrangian that governs the system in both equilibrium and non-equilibrium states.

$\cdot$ We present a case study in a privacy-preserving data collection system under local differential privacy (LDP)~\cite{duchi2013local, du2021collecting, fu2023collecting} where a non-deterministic utility function is adopted. Two strategies are devised from this analytical model, namely, Tit-for-tat and Elastic, based on which we apply the Euler-Lagrange equation to derive the system's steady-state solution.

$\cdot$ We conduct extensive experiments across varied scenarios using diverse real-world datasets to validate the effectiveness and accuracy of our proposed method.

The rest of this paper is organized as follows. Section \ref{Preliminaries} provides an introduction to the least action principle and the Euler-Lagrange function. Section \ref{Modelformulation} presents the game-theoretical model of the data collection game. Section \ref{Endless Collection Game} constructs the analytical model of the infinite collection game. Section \ref{Non-deterministic Utility} discusses the scenario where the system has a non-deterministic utility function. Section \ref{Experiments} shows the experimental results and Section \ref{Relatedwork} reviews the related work. Finally, we conclude this work in Section \ref{Conclusion}. 
\section{Preliminaries}
\label{Preliminaries}
\subsection{Generalized Coordinates}
In classical mechanics, the state of a physical system is often described by specifying its position and velocity. However, when dealing with complex systems or those subject to constraints, it's beneficial to use generalized coordinates~\cite{lagrange1853mecanique}, which are parameters that can describe the configuration of a system using the minimum number of independent variables.

When we refer to a system with $s$ degrees of freedom, we mean the system can move in $s$ independent ways. For example, a simple pendulum has one degree of freedom: the angle from the vertical axis. A double pendulum, however, has two degrees of freedom: the angle of each arm.

The use of generalized coordinates $(q_1, q_2,..., q_s)$ allows us to express the system's configuration independently of the choice of a coordinate system. This is particularly useful in dealing with systems where Cartesian coordinates are not convenient. The velocities $(\dot{q}_1, \dot{q}_2,..., \dot{q}_s)$ are the time derivatives of the generalized coordinates and represent how fast each coordinate is changing.

The trajectory of a system in its configuration space is a path that describes how the generalized coordinates change over time. It is analogous to a storyline of the system's motion, and describes where it is and how it moves at every moment in time. The system can be further described using the Lagrangian $\mathcal{L}(q_1, q_2,..., q_s,\dot{q_1}, \dot{q_2},..., \dot{q_s}, t)$. The action, a function of $\mathcal{L}$, is the integral of the trajectory, i.e., $S = \int_{t_1}^{t_2} \mathcal{L}(\boldsymbol{q},\dot{\boldsymbol{q}},t) dt$.
\subsection{The Least Action Principle}
The motion laws of a mechanical system can be expressed using the least action principle. It states that out of all possible paths that a system can take, the actual path is the one that minimizes the action. This requires that the first-order variation of the action with respect to the generalized coordinates and velocities equals zero:
\begin{equation}
\delta S = \delta \int_{t_1}^{t_2} \mathcal{L}(q_1, q_2,..., q_s,\dot{q_1}, \dot{q_2},..., \dot{q_s}, t) dt = 0,
\label{leastaction}
\end{equation}
where $\delta$ is the variation of S, i.e., an infinitely small incremental change to a function. The solution to equation \ref{leastaction} yields the Euler-Lagrange equation, whose proof can be found in any classical mechanics textbook:
\begin{lemma}
(Euler-Lagrange Equation). A necessary condition for $\delta S=0$ is:
\begin{equation}
\frac{\partial \mathcal{L}}{\partial q_i} - \frac{d}{dt}\left(\frac{\partial \mathcal{L}}{\partial \dot{q_i}}\right) = 0, \quad i=1,2,...s.
\label{ELeq}
\end{equation}
\end{lemma}

The Euler-Lagrange equation is a set of $s$ second-order ordinary differential equations, which govern the motion of the system. The equation describes how the system evolves over time, given the initial conditions of the generalized coordinates and velocities. 
\section{Game-Theoretic Model Formulation}
\label{Modelformulation}
\subsection{Threat Model}
\textbf{Attack Model.} We assume that the attacker possesses an equivalent level of information as the data collector. This implies that the attacker has full knowledge of the strategy employed by the data collector in the previous round, for example, the data collector's trimming positions. The attacker is also in agreement with the data quality standards set by the data collector, and they are acutely aware of how the poison values they send are treated. In other words,  we adopt a \textbf{white-box attack} as our attack model, which corresponds to a game with complete information. Conversely, should the attacker lack the capability to ascertain the data collector's strategy and data quality standards from the previous round, it would result in an asymmetric information scenario between the attacker and the data collector. This scenario is aligned with a black-box attack model and a game of incomplete information, which is beyond the scope of this paper. 

\textbf{Defensive Goal.} Our game-theoretic model aims to counteract a general malicious threat model where attackers are colluding, opportunistic, and \textbf{evasive}. The term ``colluding" refers to Sybil attacks, in which attackers can coordinate and share strategies to orchestrate their poison values. This situation is plausible, as these attackers may originate from a single botnet launched by one adversary. ``Opportunistic" describes attackers whose goal is to maximize the deviation of estimated statistics from the ground truth, manipulating poison values to their advantage. ``Evasive" pertains to attackers who are consistently rational, knowledgeable, and skilled enough to evade existing countermeasures by manipulating the poison value distribution \cite{cheu2021manipulation}. We believe this threat model is more comprehensive (and thus more realistic) than all existing models that commonly restrict their attacking strategies or assume that the collector has any apriori knowledge of these strategies.
\subsection{Payoff Functions}Assuming a publicly recognized data quality standard denoted by $Quality\_Evaluation()$, we establish payoff functions for both parties within the context of data manipulation attacks. Equipped with this standard, the collector can assess the intensity of poison values based on the data provided by the adversary and further determine the subsequent strategy. The existence of this metric is necessary for building up a game-theoretic model. Using this standard, let $P$ denote the payoffs for poisoning and $T$ for trimming. The game between the collector and the adversary is a zero-sum game where any gain for the adversary implies a loss for the collector and vice versa, i.e., $P_{collector} = -P_{adversary}$. However, the collector also incurs loss of accuracy due to incorrectly trimming honest values, denoted by $-T$. Hence, the collector's payoff function is $(-P-T)$.

\subsection{Strategy Space}
\subsubsection{Single Poison Value Case}
This subsection discusses the strategy space for both parties. In the single value case, where the adversary injects only one poison value, their strategy is denoted by the injection point. Similarly, the collector's strategy is determined by a trimming point in the input domain. Thus, the strategy space is represented by a pair of values $(x_{\text{adversary}}, x_{\text{collector}})$ in the input domain.

Rational players do not randomly choose strategy points from the entire space. Trimming incurs loss of utility by removing benign values, while the loss from poison values increases with more malicious data injected. However, the trimming overhead decreases as more data points are removed, making the collector more cautious when trimming.

\begin{figure}
\centering
 \subfigure[$x_L$ at the balance point]{\includegraphics[width=0.22\textwidth]{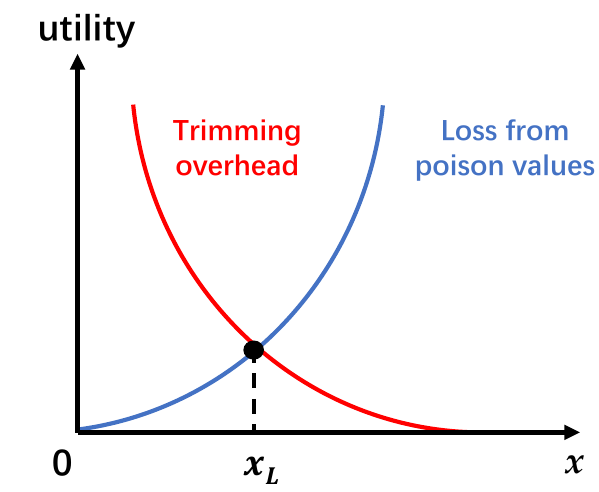}}
    \subfigure[a mixed strategy point]{\includegraphics[width=0.22\textwidth]{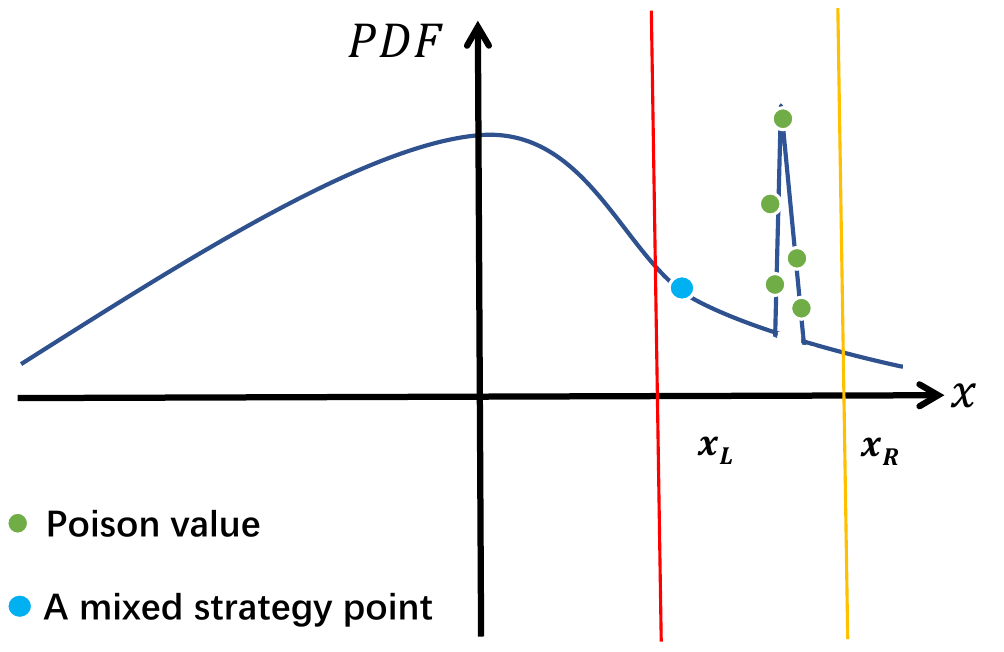}}
\caption{The definition of $x_L$, and arbitrary poison value distributions represented by a mixed strategy point}
\label{fig:defofxl}
\vspace{-5mm}
\end{figure}

As shown in Fig. \ref{fig:defofxl} (a), there exists a tradeoff between the loss caused by poison values and the overhead caused by trimming. A balance point, denoted by $x_L$, is present, such that $P(x_L) = T(x_L)$. This balance point is where the payoff for the collector and the adversary is equal, and below which the collector is not motivated to trim the data any further. In other words, a rational collector would only trim the data up to a certain point where the benefits of trimming outweigh the costs, and below that point, she would accept the risk of data poisoning to retain the accuracy. In contrast, the collector evaluates the largest acceptable value, beyond which she will definitely remove any values so that any rational adversary will not inject poison values outside of that point. As shown in Fig. \ref{fig:xlandxr}, let $x_R$ denote the maximum value that, according to the collector's belief, the adversary is willing to inject. Therefore, we have:

\begin{definition}
Let $[x_L, x_R]$ be the domain of poison values. We say the adversary plays soft if he injects poison values near $x_L$ and gains $\underline{P}$, and he plays hard if he injects poison values near $x_R$ and gains $\overline{P}$. Conversely, the collector plays soft if she trims near $x_R$ and gains $-\underline{T}$, and she plays hard if she trims near $x_L$ and gains $-\overline{T}$.
\end{definition}

\begin{figure}
\centering
\includegraphics[width=0.45\columnwidth]{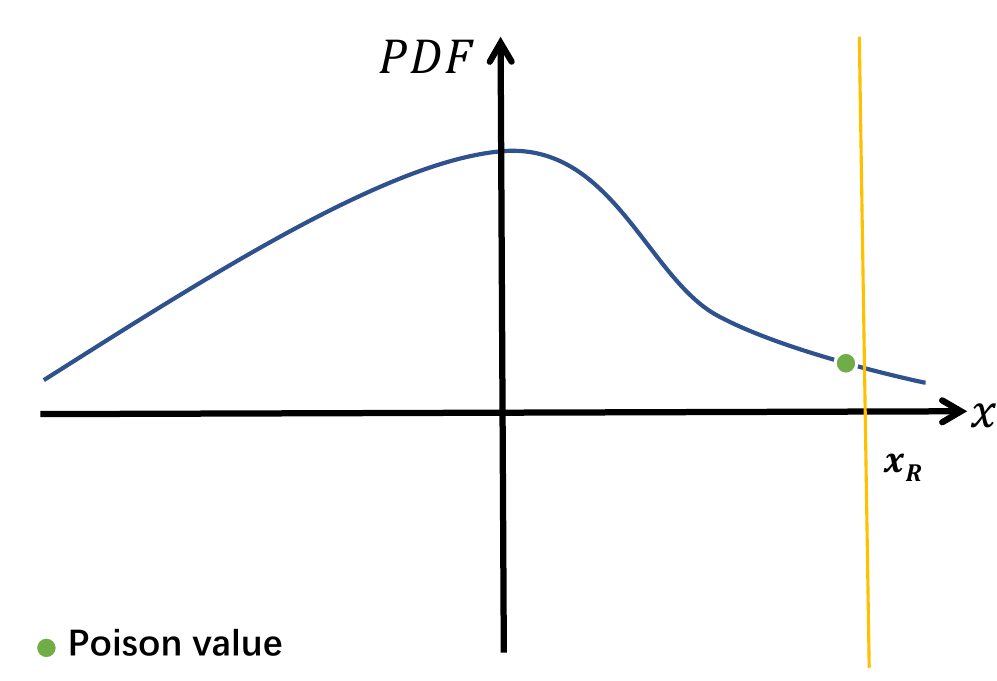}
%\hspace{-5mm}
\includegraphics[width=0.45\columnwidth]{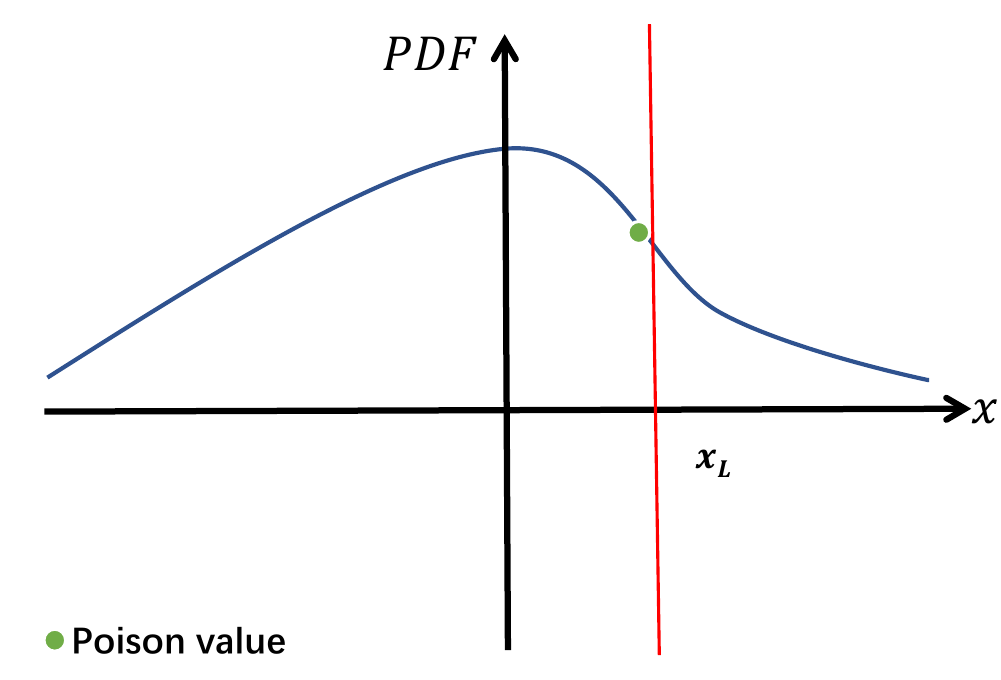}
\caption{Definition of $x_L$ and $x_R$ for a single poison value}
\label{fig:xlandxr}
\vspace{-4mm}
\end{figure}

Let the $(x_c, x_a)$ pair denote the strategies chosen by both parties, where $x_c$ denotes the trimming point, and $x_a$ denotes the point at which the adversary decides to inject poison values, and both $x_c$ and $x_a$ fall in the domain $[x_L, x_R]$. It is important to note that the strategy space is complete for both the collector and the adversary, i.e., any strategy in the domain can be chosen by both parties.

\subsubsection{General Case}
Now we discuss the general case where any poison value distribution can be deployed. Without loss of generality, any point $x_p$ in the domain $[x_L, x_R]$ can be represented as a linear combination of $x_L$ and $x_R$, i.e., there exists $p_L$ and $p_R$ such that $x_p=p_Lx_L+p_Rx_R$. This can also be viewed as a \textbf{mixed strategy} in the sense of game theory, that is, the player chooses to play $x_L$ with probability $p_L$, and $x_R$ with probability $p_R$. As such, any single poison value injection strategy over $[x_L, x_R]$ can be reduced to a mixed strategy represented by $p_Lx_L+p_Rx_R$.

Since all factors in this linear combination of $x_L$ and $x_R$ are additive, any poison value can be reduced to a single point in the strategy space. As illustrated in Fig. \ref{fig:defofxl} (b), we assert that any poison value distribution defined on $[x_L, x_R]$ can be reduced to a mixed strategy of a single poison value. As such, the strategy space for both the collector and the adversary is complete in this general case.

\subsection{Sequential Moves}In a scenario where the data collection process consists of a single round, it embodies a strategic game. Here, both the attacker and the defender simultaneously select their strategies, resulting in a straightforward and uncomplicated Nash equilibrium. As depicted in Table \ref{payoffs}, this situation mirrors the prisoner's dilemma, culminating in a unique equilibrium wherein both the adversary and the player opt for a tough stance, despite a gentler approach being mutually beneficial.

However, in real-world applications, data collection tends to be a continuous, multi-round process. This can be represented as a Stackelberg game, characterized by sequential moves where one participant's actions follow those of the other. This structure fosters cooperation, given that players can retaliate against defection, particularly when the number of rounds is indefinite or unknown. By opting for a gentler approach rather than a tough one, players can achieve a globally optimal state. The intricacies of this infinite game will be examined and modeled in detail in the following section.
\begin{table}[]
\centering
%\large
\caption{The payoff matrix of the ultimatum game, $\overline{P}>\overline{T}>>\underline{P}>\underline{T}>0$}
\renewcommand{\arraystretch}{1.5}
\begin{tabular}{c|ccc|}
\cline{2-4}
                                                 & \multicolumn{3}{c|}{Adversary}                                                               \\ \hline
\multicolumn{1}{|c|}{\multirow{3}{*}{Collector}} & \multicolumn{1}{c|}{}     & \multicolumn{1}{c|}{Soft}                 & Hard                 \\ \cline{2-4}
\multicolumn{1}{|c|}{} & \multicolumn{1}{c|}{Soft} & \multicolumn{1}{c|}{$(-\underline{P}-\underline{T}, \underline{P})$} & $(\overline{P}-\underline{T}), -\overline{P}$ \\ \cline{2-4}
\multicolumn{1}{|c|}{}                           & \multicolumn{1}{c|}{Hard} & \multicolumn{1}{c|}{$(-\overline{T}, 0)$} & $(-\overline{T}, 0)$ \\ \hline
\end{tabular}
\label{payoffs}
\vspace{-4mm}
\end{table} 
\section{Infinite Collection Game}
\label{Endless Collection Game}

\subsection{Overview}In the previous section, we emphasized the utmost importance of transforming the collection process into an infinite, roundwise repeated game to foster cooperation between the collector and the adversaries. When dealing with a limited-round scenario, wherein the game is confined to a specific number of rounds, denoted as $N$, adversaries may be tempted to defect in the final round, triggering a domino effect of defections from the second-to-last round backwards. To address this critical issue, the game must be ingeniously designed to encompass an infinite number of rounds, thereby ensuring continuous data collection.

Fig. \ref{fig:endlessgame} overviews such infinite game, wherein a data collector gathers data from the data stream (step \textcircled{3}), and an adversary attempts to inject poison values into the collected data along with normal users (step \textcircled{2}). A public board, accessible to the adversary, enables the collector to record the untrimmed data (step \textcircled{1}, \textcircled{6}). In each round, the collector collects and trims the same amount of data (step \textcircled{4}), and then determines the trimming threshold in the next round \textcircled{5}. 

Under this framework, each round becomes an invaluable opportunity to build trust, foster cooperation, learn from past experiences, and adapt strategies accordingly. The infinite nature of this repeated game actively promotes cooperation among game-theoretically rational players, paving the way for the emergence of trust and culminating in mutually beneficial outcomes in the long run.

\begin{figure}[h]
\centering
\includegraphics[width=\columnwidth]{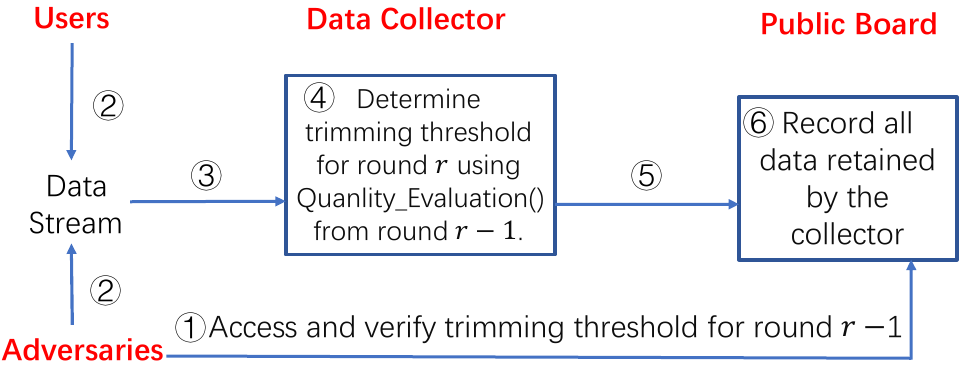}
\caption{An overview of the infinite game}
\label{fig:endlessgame}
\vspace{-5mm}
\end{figure}

\subsection{Analytical Model}
To formalize the infinite collection game, we employ the principle of least action from analytical mechanics. As the game involves an infinite number of samples collected in an infinite number of rounds, it can be viewed as a streaming process with a fixed number of samples gathered in each round. Consequently, the parameter $r$ can be regarded as a continuum, functioning as a timer within our system, analogous to the role of time $t$ in classical mechanics.

The utility functions of the adversary and the collector, $u_a$ and $u_c$, respectively, are natural coordinates that uniquely determine the state of the system. These functions are cumulative effects over $r$ and can be treated as continuous and differentiable functions of $r$. With this setting, the evolution of the system shares the same spatiotemporal structure as classical mechanics, where the general coordinate is replaced by the utility functions of both parties, and time $t$ is replaced by round $r$. We then have the fundamental principle of the infinite collection game:

\begin{axiom}
The state of the infinite collection game is determined by the least action principle, which is similar to equation \ref{leastaction}:
\begin{equation}
\delta S = \delta \int_{r_1}^{r_2} \mathcal{L}(u_a(r), u_c(r),\dot{u_a}(r), \dot{u_c}(r), r) dr = 0,
\label{leastaction2}
\end{equation}
where $\dot{u_a}=\frac{du_a}{dr}$ and $\dot{u_c}=\frac{du_c}{dr}$ are generalized velocities, and $\mathcal{L}$ is the Lagrangian.
\end{axiom}

And we have:
\begin{lemma}
The Euler-Lagrange equation of equation \ref{leastaction2} is given by:
\begin{equation}
\frac{\partial \mathcal{L}}{\partial u_a} - \frac{d}{dr}\left(\frac{\partial \mathcal{L}}{\partial \dot{u_a}}\right) = 0, \frac{\partial \mathcal{L}}{\partial u_c} - \frac{d}{dr}\left(\frac{\partial \mathcal{L}}{\partial \dot{u_c}}\right) = 0
\label{EL}
\end{equation}
\end{lemma}

\subsection{Equilibrium State}
From equation \ref{EL}, we can derive some immediate results regarding the Stackelberg equilibrium and the behavior of the collector and the adversary in the infinite collection game. When we reach a Stackelberg equilibrium, it is already a convergence state that occurs after infinite iterations of responding to each other's actions. As such, if such a convergence exists, there is no interaction between the collector and the adversary, as if they are evolving independently. Therefore, we have:

\begin{lemma}
The Lagrangian of the system is additive to that of $u_a$ and $u_c$, that is, $\mathcal{L}=\mathcal{L}(u_a)+\mathcal{L}(u_c)$.
\label{the:additive}
\end{lemma}

Since the Lagrangian should keep the form unvaried with respect to translation of $r$ and $u$, that is, we should arrive at the same Stackelberg equilibrium for any outset we choose for $r$ and $u$. From this, we have:

\begin{theorem}
$\mathcal{L}=\mathcal{L}(\dot{u}^2)$ and $\dot{u}=constant$ for any Stackelberg equilibrium state.
\label{the:NEproperty}
\end{theorem}
\begin{proof}
As $\mathcal{L}$ is only a function of the magnitude of $\dot{u}$ and is independent of the direction of $\dot{u}$, we have $\mathcal{L}=\mathcal{L}(\dot{u}^2)$. Since the Lagrangian is uniform with respect to both $r$ and $u$, it can only be an explicit function of $\dot{u}$, i.e., $\mathcal{L}=\mathcal{L}(\dot{u})$. Substitute this into equation \ref{EL}, we have $\frac{\partial \mathcal{L}}{\partial u_i} = 0$. The Euler-Lagrange function can then be written as $\frac{d}{dr}\left(\frac{\partial \mathcal{L}}{\partial \dot{u_i}}\right) = 0$. As $\frac{\partial \mathcal{L}}{\partial \dot{u_i}}$ is only a function of $u$, we have $\dot{u}=constant$. This completes the proof.
\end{proof}

Finally, we attain the form of the Lagrangian for the Stackelberg equilibrium state:

\begin{theorem}
The Lagrangian of any Stackelberg equilibrium state can be written as $\mathcal{L}=m_a\dot{u_a}^2+m_c\dot{u_c}^2$, where $m_a$ and $m_c$ are two factors regarding the adversary and the collector.
\end{theorem}
\begin{proof}
Consider an infinitesimal increment $\delta\dot{u}$ of $\dot{u}$ in the Lagrangian $\mathcal{L}$. According to Theorem \ref{the:NEproperty}, it corresponds to a Lagrangian of the form
\begin{equation}
\mathcal{L}'=\mathcal{L}((\dot{u}+\delta\dot{u})^2).
\end{equation}
Expanding it as a power series in terms of $\delta\dot{u}$ and neglecting higher order terms, considering Lemma \ref{the:additive}, we have:
\begin{equation}
\mathcal{L}((\dot{u}+\delta\dot{u})^2)=\mathcal{L}(\dot{u}^2)+2\frac{\partial\mathcal{L}}{\partial\dot{u}^2}\dot{u}\delta\dot{u}.
\end{equation}
According to Theorem \ref{the:NEproperty}, at the Stackelberg equilibrium state, $\dot{u}$ is constant. Choosing different values of $\dot{u}$ as the origin yields Euler-Lagrange equations with the same form. This means that the difference in their Lagrangians, $2\frac{\partial\mathcal{L}}{\partial\dot{u}^2}\dot{u}\delta\dot{u}$, must be a total derivative with respect to $r$. Therefore, when substituted into equation \ref{leastaction2}, $\frac{d}{dr}2\frac{\partial\mathcal{L}}{\partial\dot{u}^2}\dot{u}\delta\dot{u}$ can be eliminated in $\delta S=0$, resulting in the same Euler-Lagrange equation. Hence, $2\frac{\partial\mathcal{L}}{\partial\dot{u}^2}\dot{u}\delta\dot{u}$ and $\dot{u}$ are linearly dependent. It follows that $\frac{\partial\mathcal{L}}{\partial\dot{u}^2}$ is independent of velocity, therefore we have:
\begin{equation}
\mathcal{L}=m\dot{u}^2/2,
\end{equation}
where $m$ is a proportionality constant related to the intrinsic properties of the system. Since there are two parties, according to Lemma \ref{the:additive}, we can express the overall Lagrangian as:
\begin{equation}
\mathcal{L}=m_a\dot{u_a}^2/2+m_c\dot{u_c}^2/2.
\end{equation}
This completes the proof.
\end{proof}

Referring to the intrinsic factors associated with the attributes of both parties, it is important to acknowledge that $m_a$ and $m_c$ serve merely as two logical concepts employed to depict the system's converged state. Remarkably, they are not necessary for the determination of our strategy, specifically the trimming threshold, as will be evident in the forthcoming derivations.

\subsection{Non-equilibrium State}
A system is in a non-equilibrium state if there is a permanent non-zero interaction between the collector and the adversary, where they continuously respond to each other's last response. To mathematically describe this interaction, a term $U(u_a, u_c)$ is added to the Lagrangian, which is a function of the positions $u_a$ and $u_c$ of the collector and the adversary, respectively. Therefore, where interaction exists between the collector and the adversary, the Lagrangian can be written as:
\begin{equation}
\mathcal{L}=m_a\dot{u_a}^2+m_c\dot{u_c}^2+U(u_a, u_c)
\label{Lagrangian}
\end{equation}

The interaction term $U(u_a, u_c)$ objectively reflects the response strength of a particular strategy in relation to deviations in data quality within the practical context. It quantifies the interaction effect between the user's action and the counteraction, depending on the scenario in which it is applied. In the upcoming section, we will derive the differential equation of the infinite game according to a given form of $U$. 
\section{Non-deterministic Utility}
\label{Non-deterministic Utility}
The Tit-for-tat strategy in game theory mirrors an opponent's previous action in repeated games, fostering cooperation by rewarding cooperation and punishing defection. In Section \ref{Modelformulation}, we assume a commonly acknowledged data quality norm for both parties. However, in some practical scenarios, the utility function of a data collection system may be non-deterministic, meaning the system's outcome cannot be predicted with certainty even with known inputs.\footnote{This often occurs in privacy-preserving systems using LDP for data collection, where participants add random noise to their data before sharing. While protecting sensitive information, the noise renders the system outcome probabilistic.} Directly applying Tit-for-tat to data collection could inadvertently trigger early termination of data exchange due to the probabilistic nature of data quality assessment. This vulnerability is inherent when using Tit-for-tat in its pure form. To avoid early termination, we propose the Elastic strategy as a variant of Tit-for-tat tailored for systems with uncertain outcomes.

It should be noted that numerous variants of Tit-for-tat exist, such as Tits-for-two-tats~\cite{axelrod1981evolution} and Generous Tit-for-tat~\cite{nowak1992tit}. They can also be adapted through Elastic strategies for repeated games with uncertainty. For simplicity, this paper focuses the discussion on the original Tit-for-tat. The insights can be readily extended to other variants of Tit-for-tat.

\subsection{Tit for Tat Strategy}
The data collector selects the following parameters: $Tth$, the trimming threshold; \textit{Quality\_Evaluation()}, which measures the quality of the data $X_i$ received in the $i$-th round; $Round\_no$, which represents the number of data collection rounds; $Quality\_Evaluation(X_0)$, the triggering condition; and \textit{Red}, a redundancy to ensure that the termination round is not too small. The procedure of Titfortat is given in Algorithm \ref{alg:Titfortat}.

\begin{algorithm}[]
\caption{Titfortat Strategy}
\SetKwInput{Input}{Input}
\SetKwInput{Output}{Output}

\Input{\textit{Quality\_Evaluation()}, $X_0$, \textit{Red}, $\underline{T}$, $\overline{T}$, $Round\_no$}
\Output{$Round\_terminate$}

$Tth$ $\gets \underline{T}$;

$Round\_terminate$ $\gets$ $Round\_no$;

\For{$i \gets 1$ \KwTo Round\_no}{
\If{$\textit{Quality\_Evaluation}(X_i) < \textit{Quality\_Evaluation}(X_0) + \textit{Red}$}{
$Tth \gets \overline{T}$;

$Round\_terminate$ $\gets i$;

break;
}
}
\Return $Round\_terminate$
\label{alg:Titfortat}
\end{algorithm}

It is easy to apply when utility is deterministic. As a trigger strategy requiring permanent termination of cooperation upon betrayal, we have the interaction term $U(u_a,u_c)$ for the Tit-for-tat strategy becomes $U(u_a,u_c)=0$, if $u_a=u_c$ and otherwise $U(u_a,u_c)=\infty$. In non-deterministic utility scenarios, cooperation termination may be triggered by normal jitter even if both parties are cooperative. Intuitively, the collector should compromise their roundwise gain to preserve redundancy and maximize long-run benefit. 

In the Stackelberg equilibrium, we assume that both the collector and the adversary have a symmetric setting. This means that if $u_a$ and $u_c$ are symmetric, the solution should also be symmetric. Let $g_c=\overline{T}-\underline{P}-\underline{T}$ and $g_a=\overline{P}$ be the roundwise gain for both parties during cooperation (which is the payoff of compliance minus betrayal), and $g_{ac}=\frac{g_a+g_c}{2}$ due to the symmetry axiom. The collector now expects a gain of $g_0=g_{ac}-\delta$, where $\delta$ is a compromise in data utility. If the adversary complies, the collector can observe compliance deterministically since the probability of the data utility being less than $g_0$ in the outcome is negligible. However, if the adversary defects at $g_{ac}$, the collector judges compliance with probability $p$ and defects with probability $1-p$ due to the perturbation's probabilistic nature. With these settings, we derive the following theorem concerning the Stackelberg equilibrium for the Tit-for-tat strategy:

\begin{theorem}
The condition for the adversary choosing to comply in the Tit-for-tat game is $\delta<\frac{d-dp}{1-dp}g_{ac}$, where $d$ denotes the roundwise discount rate of data utility acknowledged by both parties.
\end{theorem}

\begin{proof}
From the adversary's perspective, their current-round gain expectation when choosing to comply is 
\begin{equation}
g_{com}=g_{0}+dg_{com},\quad \text{or}\quad g_{com}=\frac{g_0}{1-d}. 
\end{equation}
However, if the adversary opts to defect, they will be assessed in the subsequent round as complying with probability $p$ and defecting with probability $1-p$. As a result, their current-round gain expectation becomes
\begin{equation}
g_{def}=g_{ac}+dpg_{def},\quad \text{or}\quad g_{def}=\frac{g_{ac}}{1-dp}.
\end{equation} 
The adversary will decide to comply if and only if $g_{com}>g_{def}$, which is equivalent to $g_0>\frac{1-d}{1-dp}g_{ac},\quad \text{or}\quad \delta<\frac{d-dp}{1-dp}g_{ac}$. This completes the proof.
\end{proof}

Should $p=1$, implying that the adversary is never identified as defecting, they would always opt to defect given the lack of consequences. In contrast, as $p\rightarrow 0$, signifying an increased likelihood of the adversary being flagged as defecting, a substantial adjustment must be made to $\delta$ to cultivate trust. This analysis underscores the complexities and trade-offs inherent in managing non-deterministic utility situations within data collection systems. The delicate balance between cooperation, trust, and data utility is crucial to the system's sustained success. Hence, given $\overline{T}, \underline{T}, \overline{P}, \underline{P}, p, d$, one can ascertain the $Tth$ of Tit-for-tat by selecting a $\delta$ according to their preference.

\subsection{Elastic Trigger Strategy}So far, we have discussed the equilibrium of the Tit-for-tat strategy, which is essentially a rigid trigger strategy. Unfortunately, while preserving redundancy effectively extends the period of cooperation, we should note that the game cannot achieve infinite rounds, as the probability of termination keeps increasing and will ultimately converge to 1 in the long run. A simple way to tackle this is to allow the trigger strategy to be elastic with forgiveness, namely, applying a penalty in the next round when a defection is detected instead of terminating the cooperation directly. This is shown in Algorithm \ref{alg:Elastic}.

\begin{algorithm}[]
\caption{Elastic Trigger Strategy}
\SetKwInput{Input}{Input}
\SetKwInput{Output}{Output}

\Input{\textit{Quality\_Evaluation()}, $\underline{T}$, $\overline{T}$, $Round\_no$, $k$}
\Output{$Tth_i$}

$Tth_1$ $\gets \underline{T}$;

$QE_i=\frac{Quality\_Evaluation(X_i)}{max(Quality\_Evaluation(\cdot))}$

\For{$i \gets 2$ \KwTo $Round\_no$}{
$Tth_i=(1-kQE_i)\underline{T}+kQE_i\overline{T}$
}
\Return $Tth_i$
\label{alg:Elastic}
\end{algorithm}
%\vspace*{-0.8cm}

That implies an interaction between the adversary and the collector exists, and an equilibrium position exists such that the interaction pulls the relative utility $|u_a-u_c|$ back to the equilibrium position by a force equal to $-\frac{\partial U}{\partial u_a}$ or $-\frac{\partial U}{\partial u_c}$. We expand this interaction into a power series about $(u_a-u_c)$. When the deviation between the two is small, $u_a-u_c$ is also small, so only the first non-zero item is reserved. This is a quadratic term, and we introduce a proportionality constant k that describes the strength of the interaction. Therefore, we have:

\begin{definition}
The interaction term $U(u_a,u_c)$ for the elastic trigger strategy is
\begin{equation}
U(u_a,u_c)=k(u_a-u_c)^2/2.
\label{Elainteraction}
\end{equation}
\end{definition}

According to this, we have the following theorem:
\begin{theorem}
The utility functions of the adversary and the collector periodically oscillate with respect to $r$ in the setting of the elastic strategy.
\end{theorem}
\begin{proof}
The Lagrangian of this system is given by 
\begin{equation}
\mathcal{L}=m_a\dot{u_a}^2+m_c\dot{u_c}^2+k(u_a-u_c)^2/2
\end{equation}
by plugging equation \ref{Elainteraction} into equation \ref{Lagrangian}. Applying equation \ref{EL} to this, we have 
\begin{equation}
m_a\ddot{u_a}+k(u_a-u_c)=0, m_c\ddot{u_c}+k(u_a-u_c)=0
\end{equation}
These two equations have the same form of ordinary differential equations concerning that of a double harmonic oscillator system, where two masses $m_a$ and $m_c$ are connected by a spring with spring constant $k$. The solution will also be the same, in the form of 
\begin{equation}
u(r)=A\cos(\omega r+\phi).
\end{equation}
This completes the proof.
\end{proof}
\section{Experiments}
\label{Experiments}
In this section, the performance of the proposed approach is evaluated through its application to real-world datasets. Experiments are implemented in MATLAB R2021b on a desktop computer with Intel i7-10700K RTX 3090 eight-core CPU, 128GB RAM, and Windows 10 OS.

\subsection{Experimental Setup}
\textbf{Datasets.} In our experiments, we use 5 real-world numerical datasets. \textbf{Control}, \textbf{Vehicle}, and \textbf{Letter}~\cite{UCI2023} are standard UCI datasets, frequently used in machine learning research. \textbf{Taxi}~\cite{taxi2018}, extracted from the January 2018 New York Taxi data, records the pick-up times during a day and includes 1,048,575 integers from 0 to 86,340, normalized to the range $[-1,1]$. \textbf{Creditcard}~\cite{OpenML2015} comprises numerical results of PCA transactions, which are sanitized to preserve confidentiality. A summary of all dataset information is shown in Table \ref{tab:dataset-info}.

\begin{table}[htbp]
\centering
%\scriptsize
\caption{Dataset Information}
\label{tab:dataset-info}
\begin{tabular}{|l|c|c|c|}
\hline
Dataset & Instances & Features & Clusters \\
\hline
\textbf{\MakeUppercase{Control}} & 600 & 60 & 6\\
\hline
\textbf{\MakeUppercase{Vehicle}} & 752 & 18 & 4\\
\hline
\textbf{\MakeUppercase{Letter}} & 20000 & 16 & 26\\
\hline
\textbf{\MakeUppercase{Taxi}} & 1048575 & 1 & 1\\
\hline
\textbf{\MakeUppercase{Creditcard}} & 284807 & 31 & 4\\
\hline
\end{tabular}
%\vspace{-4mm}
\end{table}

\textbf{Parameter Settings.} In order to standardize the description of our approach across different datasets, we describe the positions of poison value injection and trimming in terms of data percentiles. The position for each trimming round is determined by the parameter $Tth$.

We implement several benchmark schemes. Groundtruth represents the result obtained by running the original dataset without any poison value injection. Ostrich assumes no defensive measures are taken, i.e., accepting all poison values. Since the adversary is also aware of this, the poison value is injected at the 99th percentile in each round. We also implement two baseline defence schemes where the data collector sets static thresholds. In the $Baseline_{0.9}$ scheme, the adversary randomly injects poison values in the percentile range of [0.9, 1],  while in the $Baseline_{static}$ scheme, the adversary injects poison values at the percentile $(Tth-1\%)$. The latter is the ideal attack, which assumes that the adversary has the ability to accurately determine the data collector's $Tth$ for each round and always adds poison values at the location that benefits itself the most.

We implement our three proposed schemes, namely Titfortat, $Elastic_{0.1}$, and $Elastic_{0.5}$. These schemes employ different strategies for setting the trimming and injection positions for poison values, with varying levels of adaptability based on previous adversary actions. In the Titfortat scheme, the untriggered trim position is set at the $(Tth+1\%)$ percentile, but once the adversary triggers the judgement, the subsequent rounds will always be trimmed at the $(Tth-3\%)$ percentile. In the Elastic schemes, the initial trim position is set at the $(Tth-3\%)$ percentile, and the initial injection position for poison value is set at the $(Tth+1\%)$ percentile. In the subsequent rounds, the data collector dynamically adjusts the trimming position for the next round based on the previous round's adversary's poison injection position $A(i)$, according to the rule $T(i+1)=Tth+k(A(i)-Tth-1\%)$, while the adversary adjusts the poison injection position for the next round $A(i+1)$ based on the previous round's data collector's trimming position $T(i)$, according to the rule $A(i+1)=Tth-3\%+k(T(i)-Tth)$. $Elastic_{0.1}$ and $Elastic_{0.5}$ represent the parameter $k$ taking the values 0.1 and 0.5, respectively.

\subsection{Stackelberg Equilibrium Results on k-Means Clustering}This subsection presents the clustering results from k-means applied to \textbf{Control}, \textbf{Vehicle}, and \textbf{Letter}. We compare the performance when both the data collector and the adversary follow Stackelberg equilibrium strategies. For each experiment, we consider 20 rounds of games, with results averaged over 100 repetitions. Titfortat is assumed not to experience early terminations. Three attack ratio intervals are [0, 0.01], [0.05, 0.15], and [0.2, 0.5], corresponding to the situations where there are few, moderate, and many poison values, respectively. Results under different attack ratios are named after the dataset name and the corresponding attack ratio interval, for example, $\textbf{Control}_{[0,0.01]}$.

Fig. \ref{fig:NE090} illustrates the results when the Tth is set to 0.9. The y-axis in this chart depicts two distinct measurements: the Sum of Squares Errors (SSE) and Distance. SSE is defined by the equation $SSE = \sum_{i=1}^{n} (y_i - \hat{y_i})^2$, where $y_i$ represents the observed values and $\hat{y_i}$ stands for the predicted values. On the other hand, `Distance' illustrates the discrepancy between the actual centroid of the clustering and the ground truth, as measured by the Euclidean distance. We observe that during intervals of low attack ratios, the volume of poison values is minimal. As such, Ostrich performs optimally and manifests the smallest offset. In such situations, all schemes implementing trimming end up with additional overhead costs. As the attack ratio escalates, pushing into a grey area where the impact of trimming to eliminate poison values is counterbalanced by false positives, the performance of the Ostrich scheme gradually deteriorates. In contrast, when the attack ratio falls within a large interval, where poison values become dominant, our proposed schemes significantly outperform both baseline schemes. Also, it is evident that Ostrich has the highest SSE. Moreover, in almost all the cases presented, our proposed solutions outperform both baseline approaches, with $Elastic_{0.5}$ demonstrating the best performance.

The results when the Tth is adjusted to 97\% are depicted in Fig. \ref{fig:NE097}, from which similar conclusions can be drawn. Here, the trimming method adopted is more conservative, thus diminishing the overhead at lower attack ratios. However, the effectiveness of this approach becomes less distinct at higher attack ratios.

\begin{figure*}[]
    \centering
    \begin{minipage}[t]{\linewidth}
\centering
\includegraphics[width=0.8\textwidth]{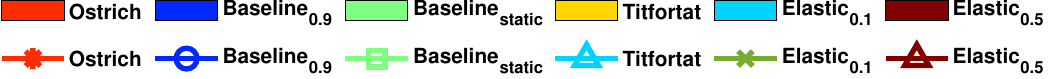}
\end{minipage}
    \subfigure[$\textbf{Control}_{[0,0.1]}$]{\includegraphics[width=0.32\textwidth]{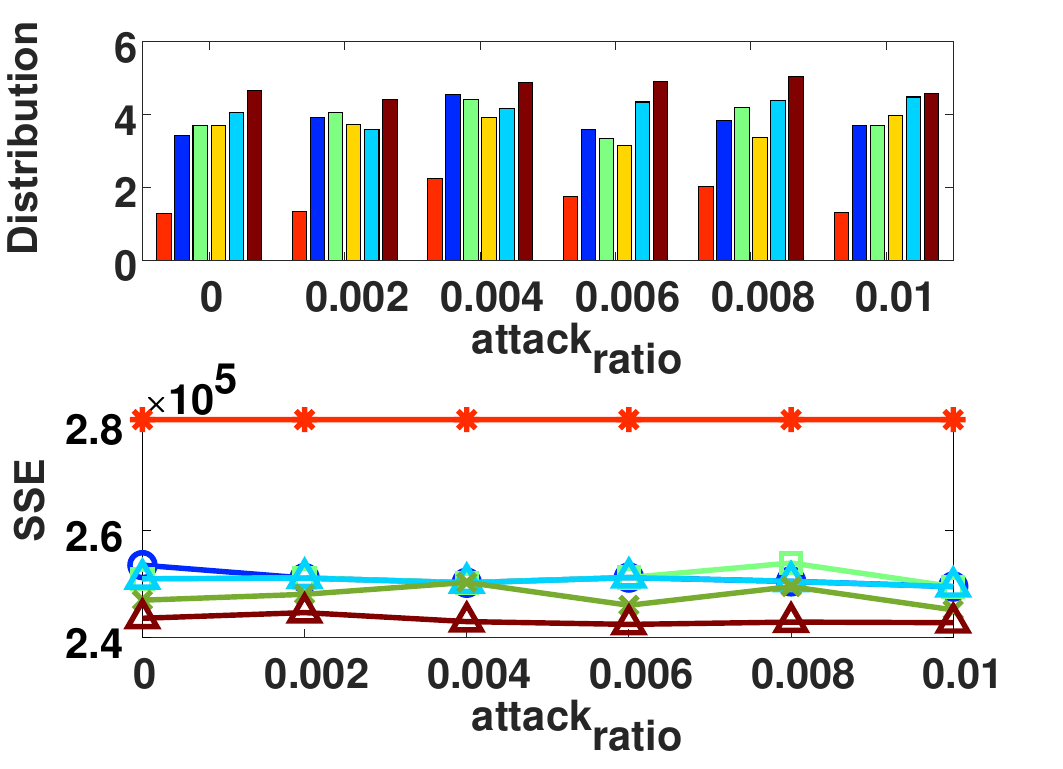}}\hspace{-0.02\textwidth}
    \subfigure[$\textbf{Vehicle}_{[0,0.1]}$]{\includegraphics[width=0.32\textwidth]{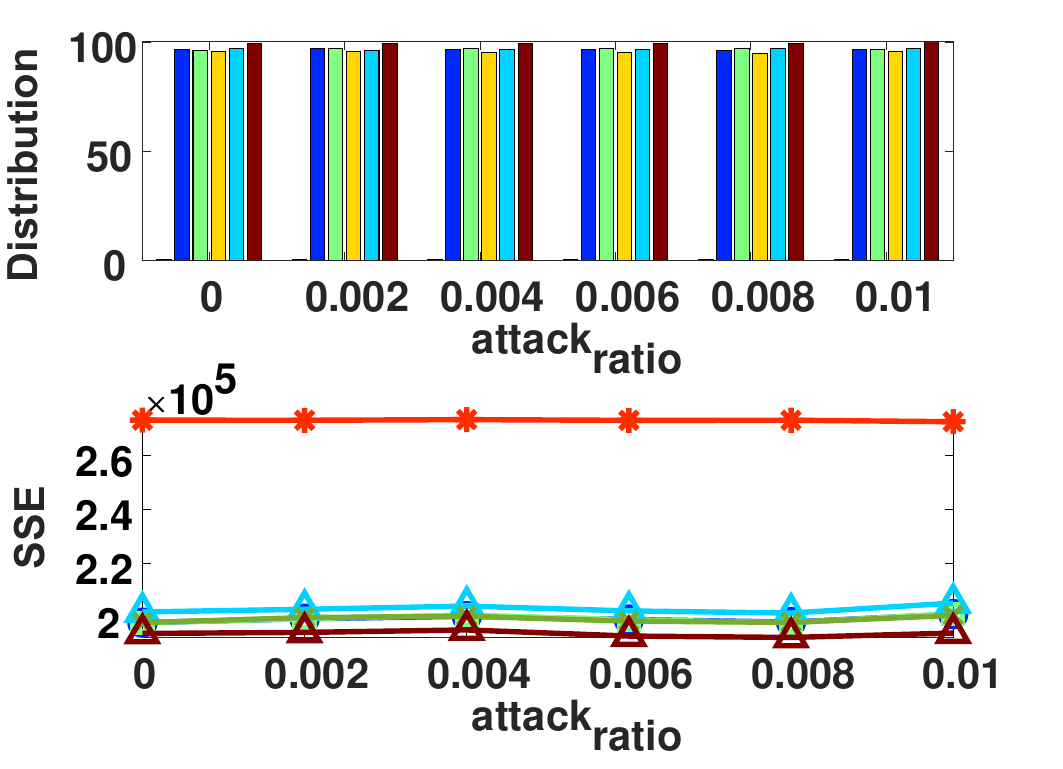}}\hspace{-0.02\textwidth}
    \subfigure[$\textbf{Letter}_{[0,0.1]}$]{\includegraphics[width=0.32\textwidth]{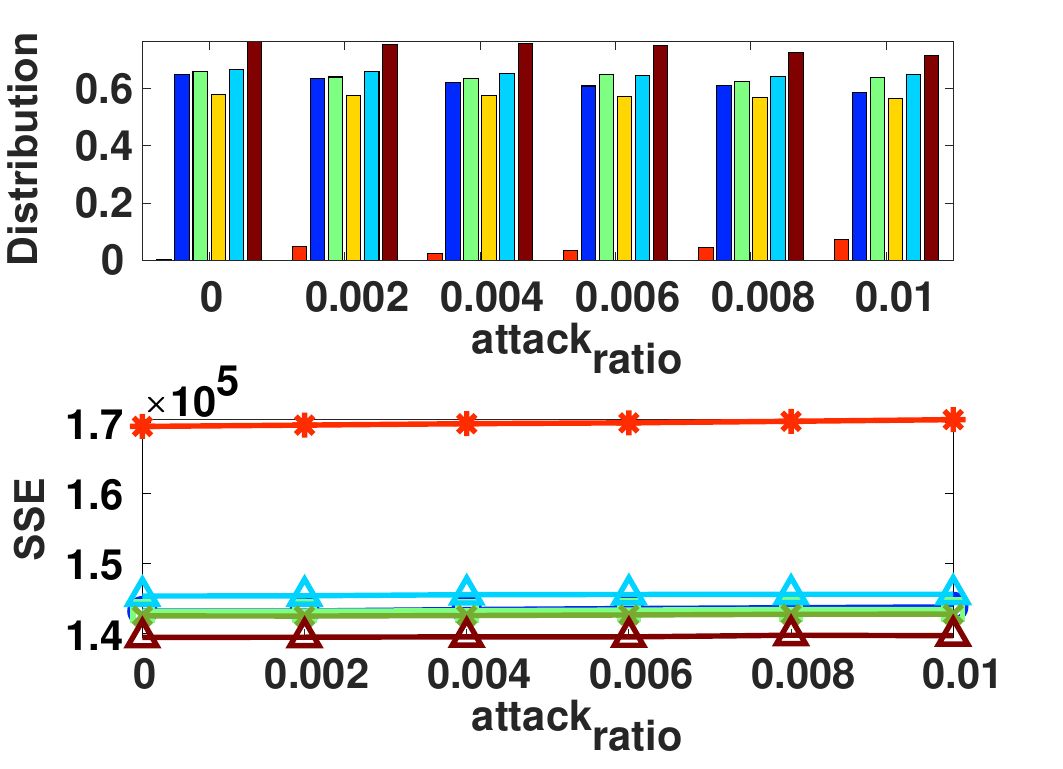}}\hspace{-0.02\textwidth}

    \subfigure[$\textbf{Control}_{[0.05,0.15]}$]{\includegraphics[width=0.32\textwidth]{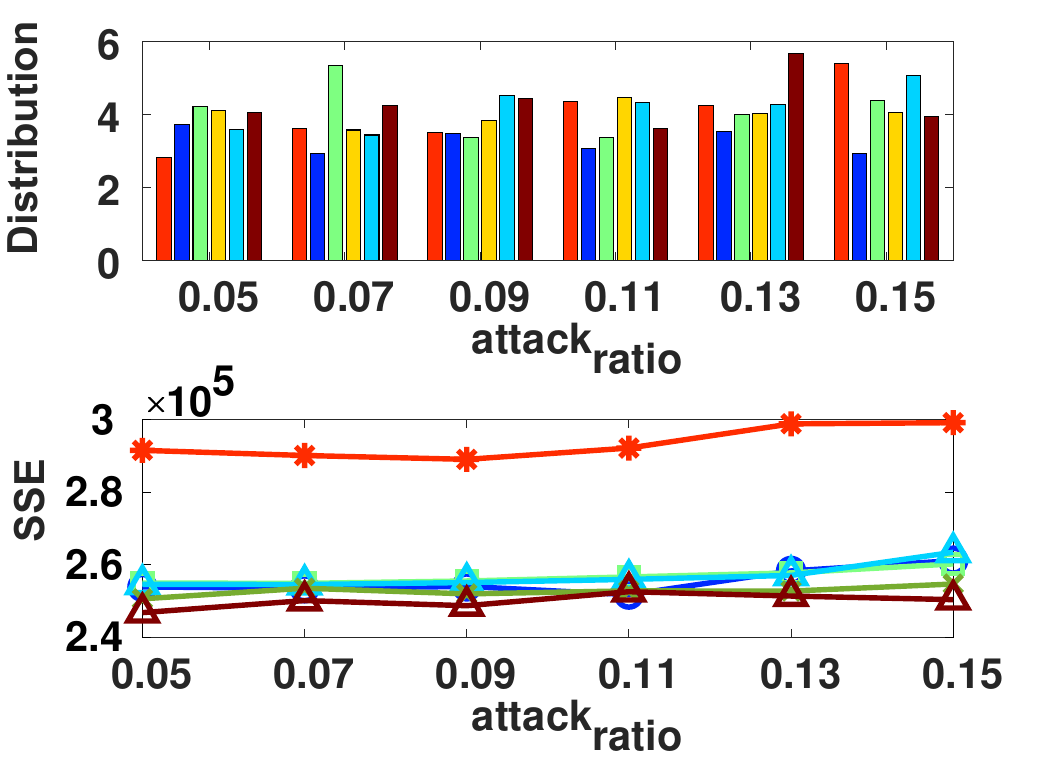}}\hspace{-0.02\textwidth}
    \subfigure[$\textbf{Vehicle}_{[0.05,0.15]}$]{\includegraphics[width=0.32\textwidth]{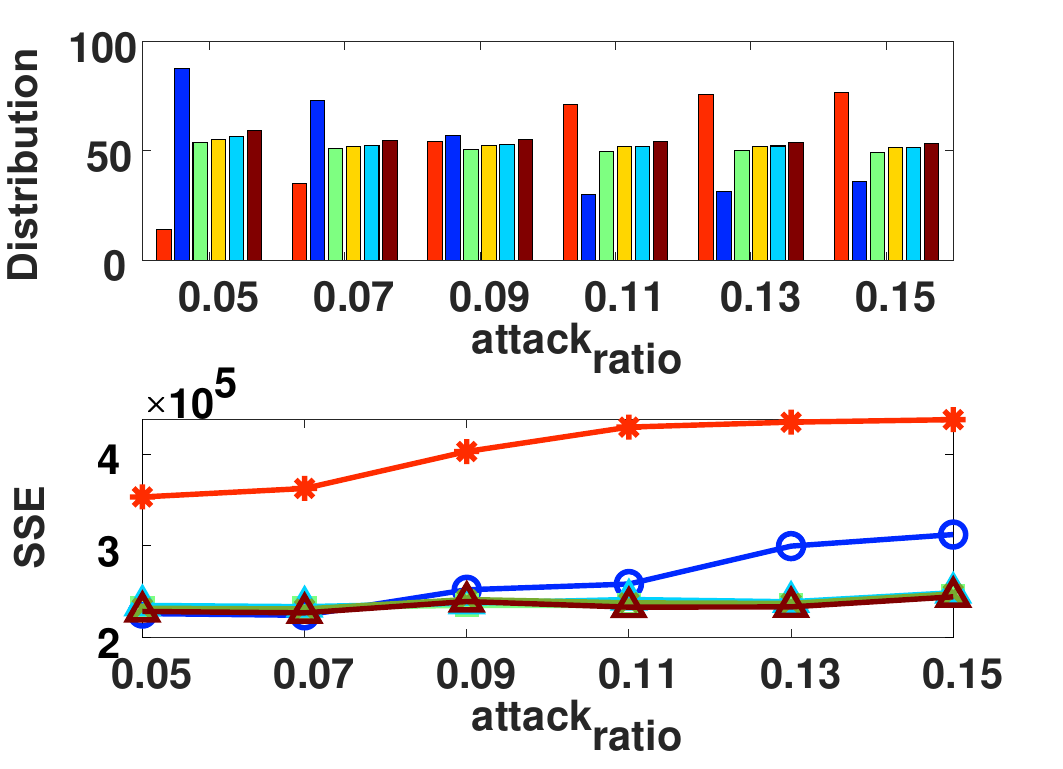}}\hspace{-0.02\textwidth}
    \subfigure[$\textbf{Letter}_{[0.05,0.15]}$]{\includegraphics[width=0.32\textwidth]{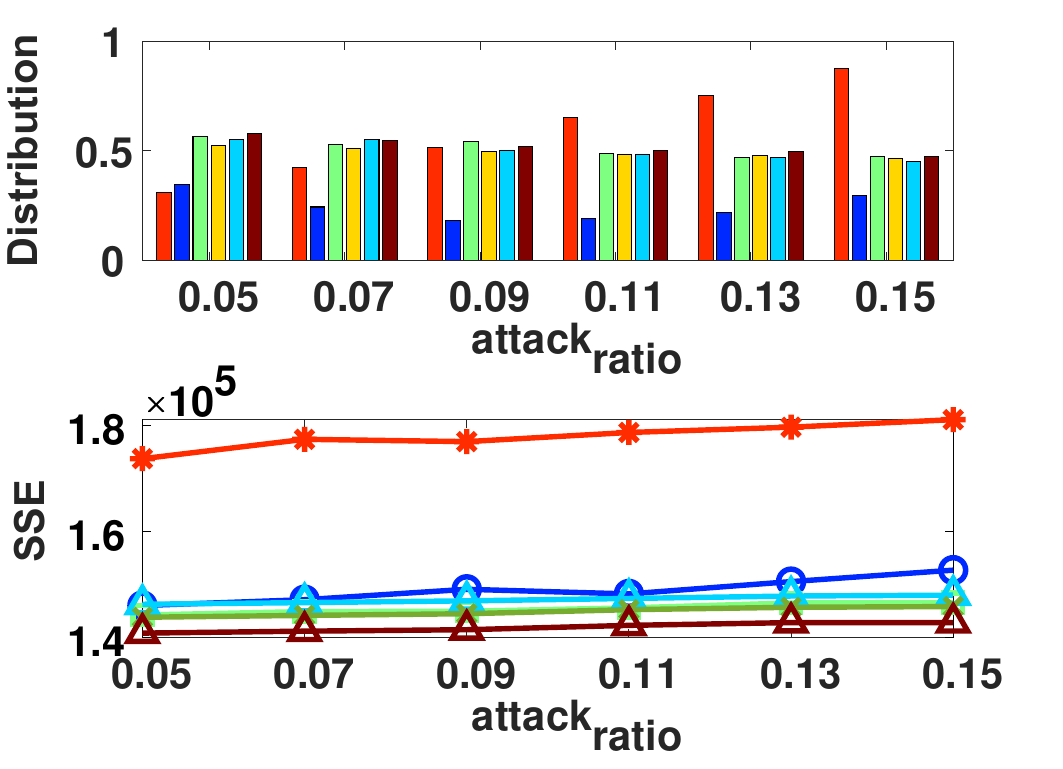}}\hspace{-0.02\textwidth}

    \subfigure[$\textbf{Control}_{[0.2,0.5]}$]{\includegraphics[width=0.32\textwidth]{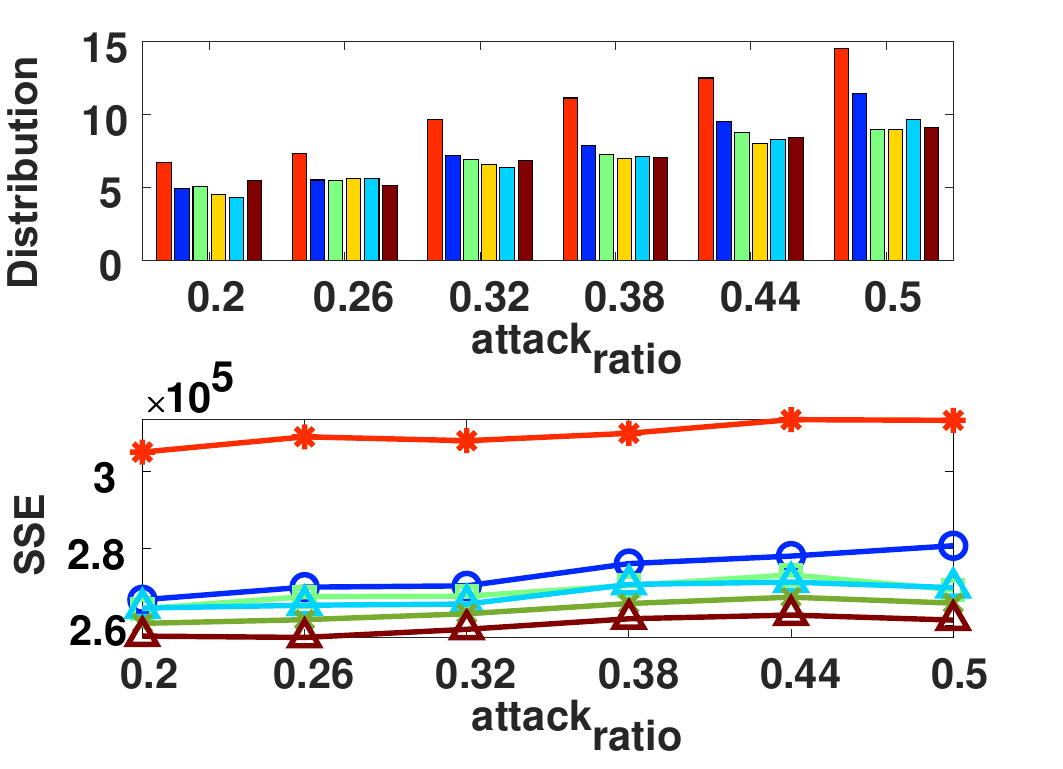}}\hspace{-0.02\textwidth}
    \subfigure[$\textbf{Vehicle}_{[0.2,0.5]}$]{\includegraphics[width=0.32\textwidth]{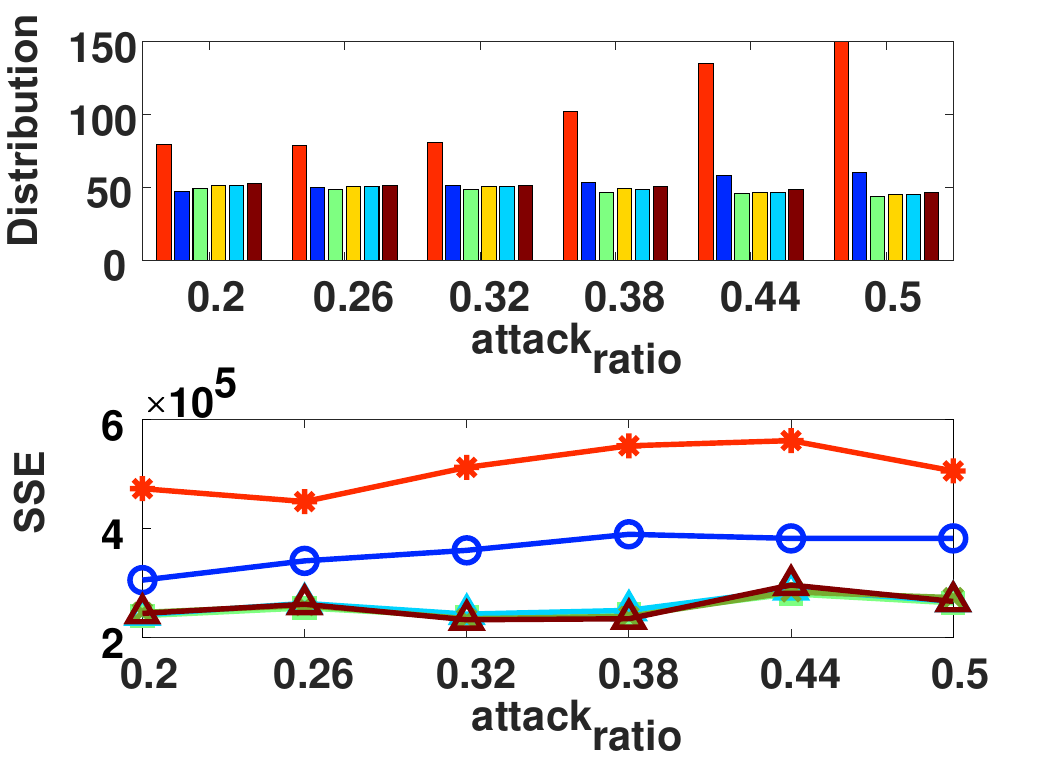}}\hspace{-0.02\textwidth}
    \subfigure[$\textbf{Letter}_{[0.2,0.5]}$]{\includegraphics[width=0.32\textwidth]{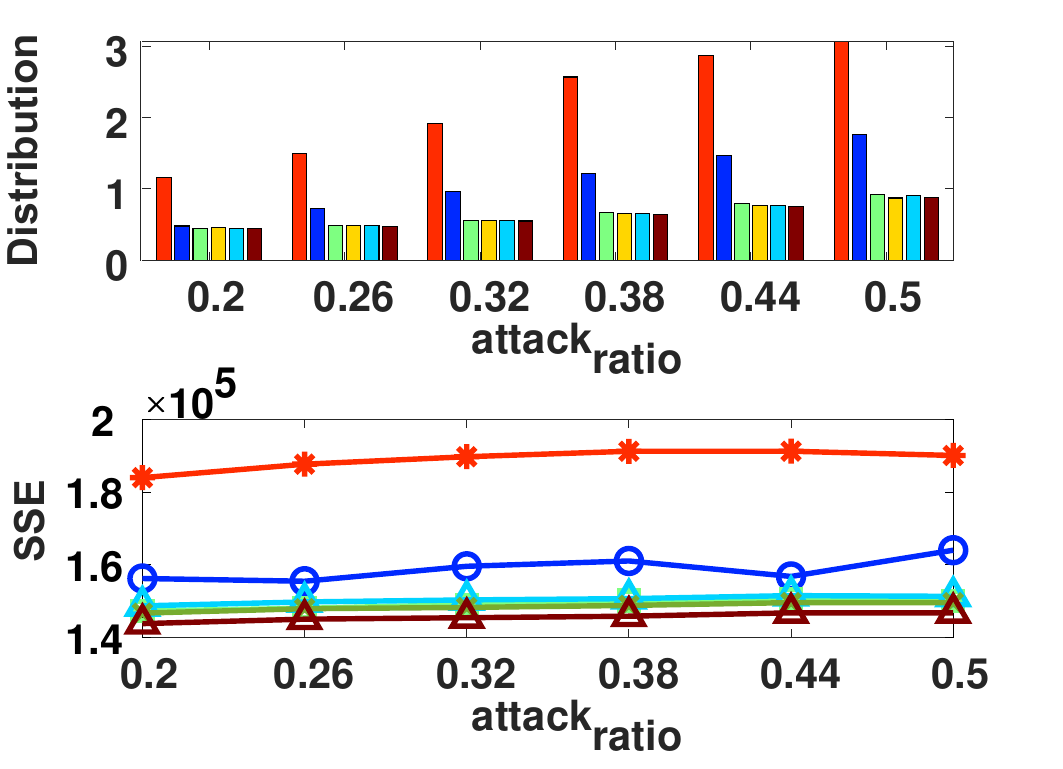}}\hspace{-0.02\textwidth}

    \caption{K-means clustering results over \textbf{Control}, \textbf{Vehicle}, and \textbf{Letter}, Tth=0.9}
    \label{fig:NE090}
    \vspace{-4mm}
\end{figure*}

\begin{comment}
\begin{figure*}[]
    \centering
        \begin{minipage}[t]{\linewidth}
\centering
\includegraphics[width=0.9\textwidth]{figures/legend.pdf}
\end{minipage}
    \subfigure[$\textbf{Control}_{[0,0.1]}$]{\includegraphics[width=0.3\textwidth]{figures/C95L.pdf}}\hspace{-0.02\textwidth}
    \subfigure[$\textbf{Vehicle}_{[0,0.1]}$]{\includegraphics[width=0.3\textwidth]{figures/V95L.pdf}}\hspace{-0.02\textwidth}
    \subfigure[$\textbf{Letter}_{[0,0.1]}$]{\includegraphics[width=0.3\textwidth]{figures/L95L.pdf}}\hspace{-0.02\textwidth}

    \subfigure[$\textbf{Control}_{[0.05,0.15]}$]{\includegraphics[width=0.3\textwidth]{figures/C95M.pdf}}\hspace{-0.02\textwidth}
    \subfigure[$\textbf{Vehicle}_{[0.05,0.15]}$]{\includegraphics[width=0.3\textwidth]{figures/V95M.pdf}}\hspace{-0.02\textwidth}
    \subfigure[$\textbf{Letter}_{[0.05,0.15]}$]{\includegraphics[width=0.3\textwidth]{figures/L95M.pdf}}\hspace{-0.02\textwidth}

    \subfigure[$\textbf{Control}_{[0.2,0.5]}$]{\includegraphics[width=0.3\textwidth]{figures/C95H.pdf}}\hspace{-0.02\textwidth}
    \subfigure[$\textbf{Vehicle}_{[0.2,0.5]}$]{\includegraphics[width=0.3\textwidth]{figures/V95H.pdf}}\hspace{-0.02\textwidth}
    \subfigure[$\textbf{Letter}_{[0.2,0.5]}$]{\includegraphics[width=0.3\textwidth]{figures/L95H.pdf}}\hspace{-0.02\textwidth}

    \caption{K-means clustering results over \textbf{Control}, \textbf{Vehicle}, and \textbf{Letter}, Tth=0.95}
    \label{fig:NE095}
\end{figure*}
\end{comment}

\begin{figure*}[]
    \centering
        \begin{minipage}[t]{\linewidth}
\centering
\includegraphics[width=0.8\textwidth]{figures/legend.pdf}
\end{minipage}
    \subfigure[$\textbf{Control}_{[0,0.1]}$]{\includegraphics[width=0.32\textwidth]{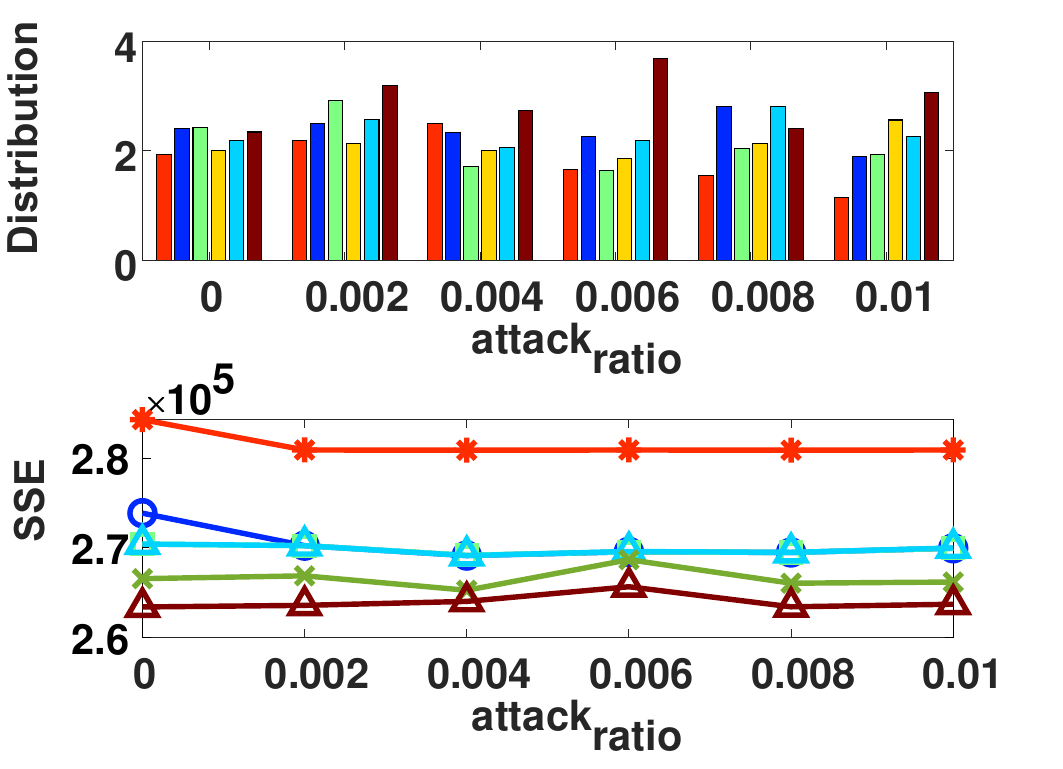}}\hspace{-0.02\textwidth}
    \subfigure[$\textbf{Vehicle}_{[0,0.1]}$]{\includegraphics[width=0.32\textwidth]{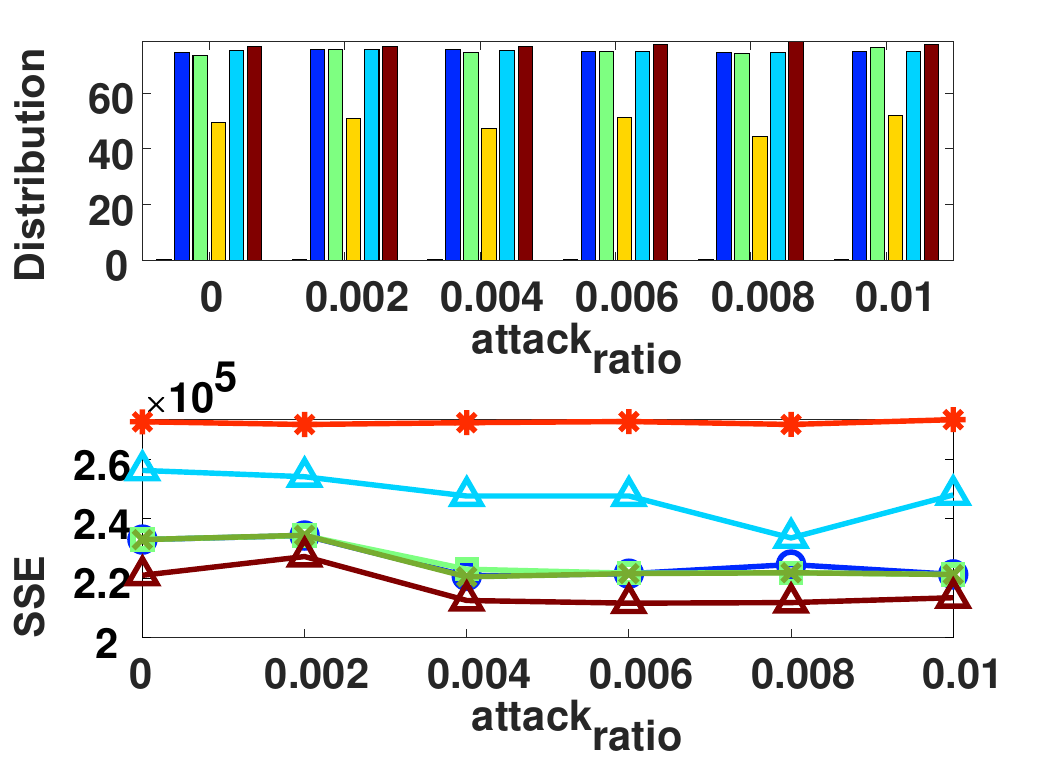}}\hspace{-0.02\textwidth}
    \subfigure[$\textbf{Letter}_{[0,0.1]}$]{\includegraphics[width=0.32\textwidth]{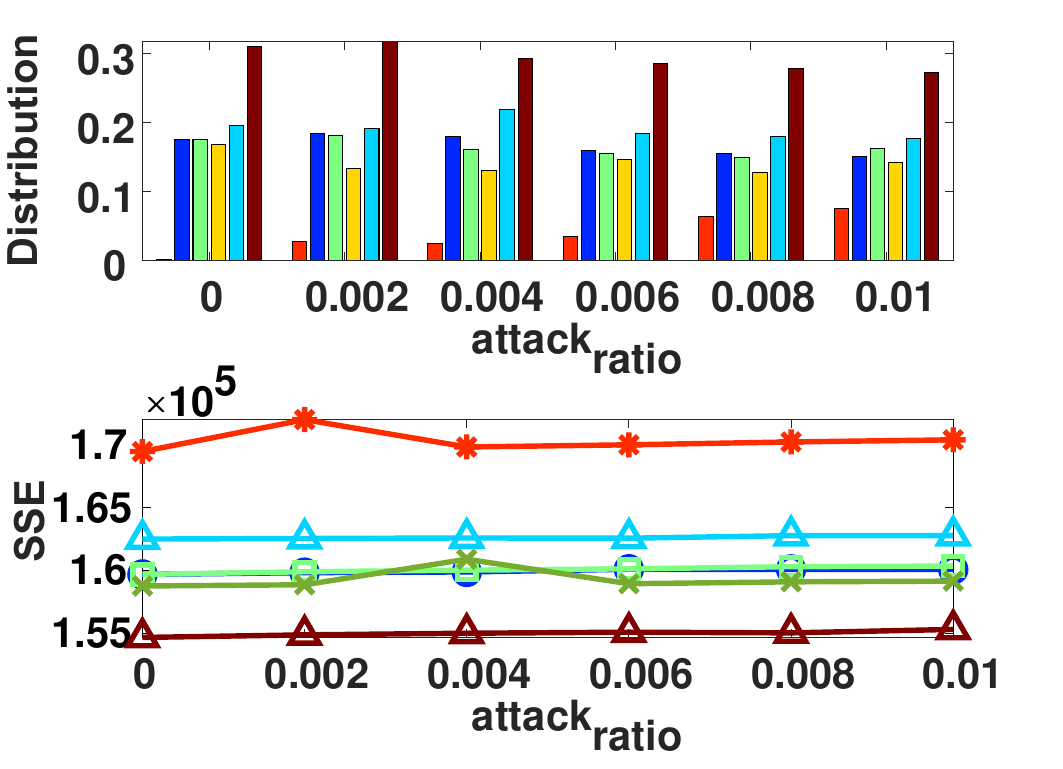}}\hspace{-0.02\textwidth}

    \subfigure[$\textbf{Control}_{[0.05,0.15]}$]{\includegraphics[width=0.32\textwidth]{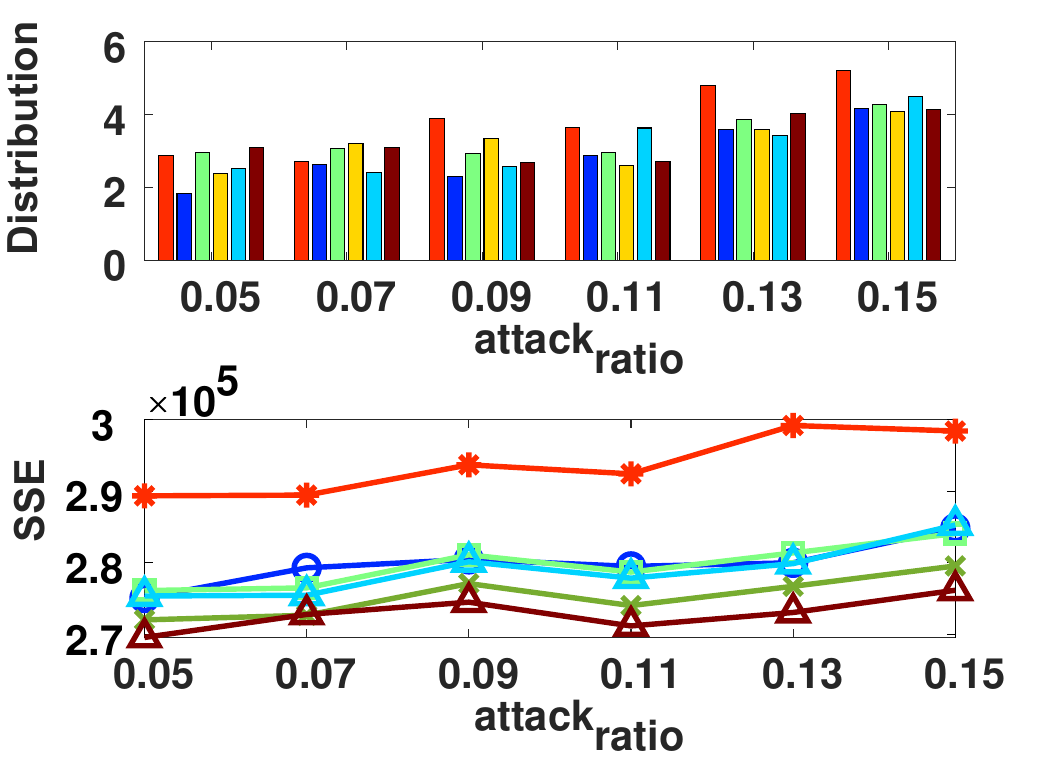}}\hspace{-0.02\textwidth}
    \subfigure[$\textbf{Vehicle}_{[0.05,0.15]}$]{\includegraphics[width=0.32\textwidth]{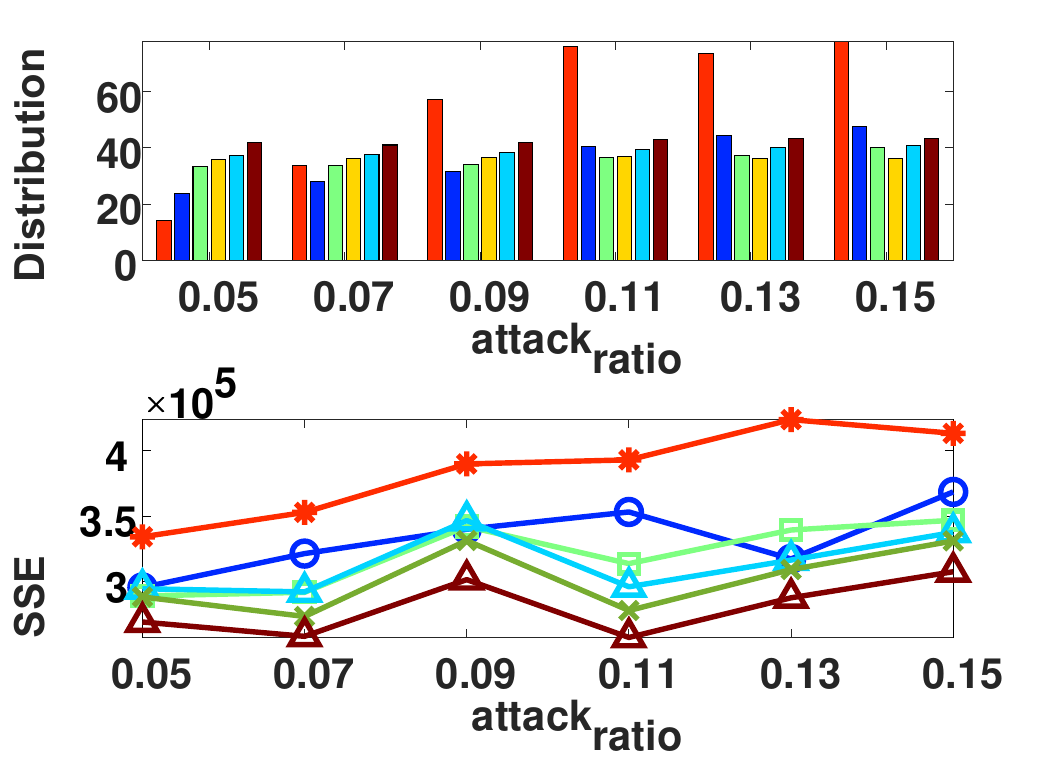}}\hspace{-0.02\textwidth}
    \subfigure[$\textbf{Letter}_{[0.05,0.15]}$]{\includegraphics[width=0.32\textwidth]{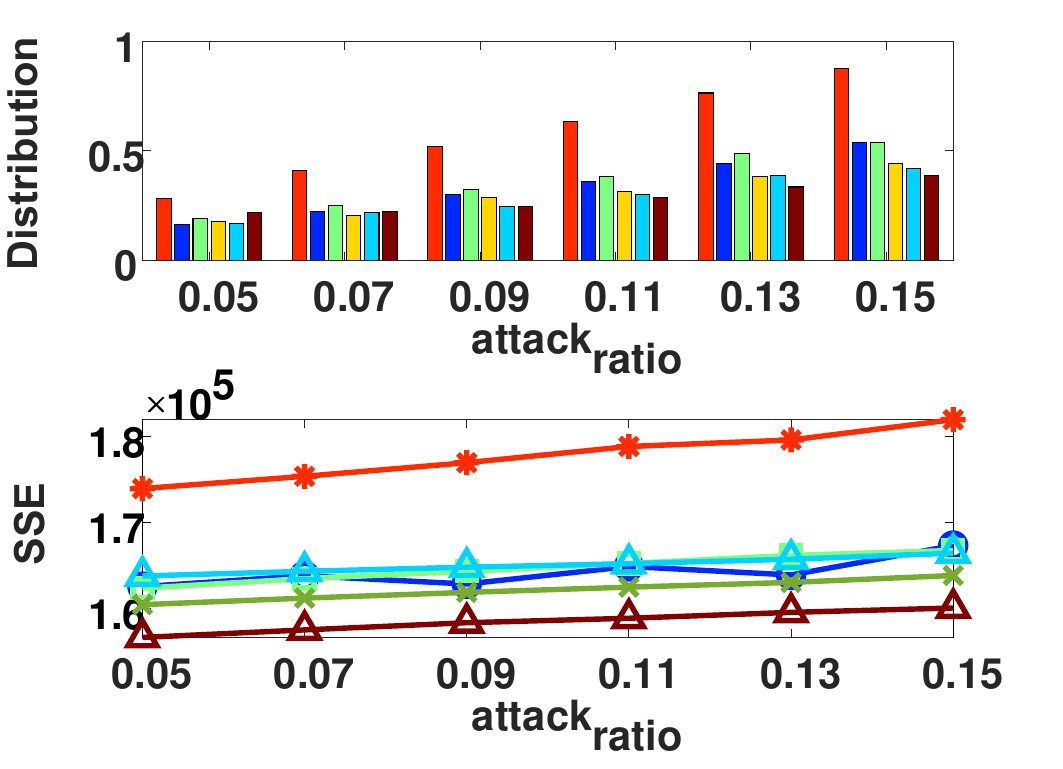}}\hspace{-0.02\textwidth}

    \subfigure[$\textbf{Control}_{[0.2,0.5]}$]{\includegraphics[width=0.32\textwidth]{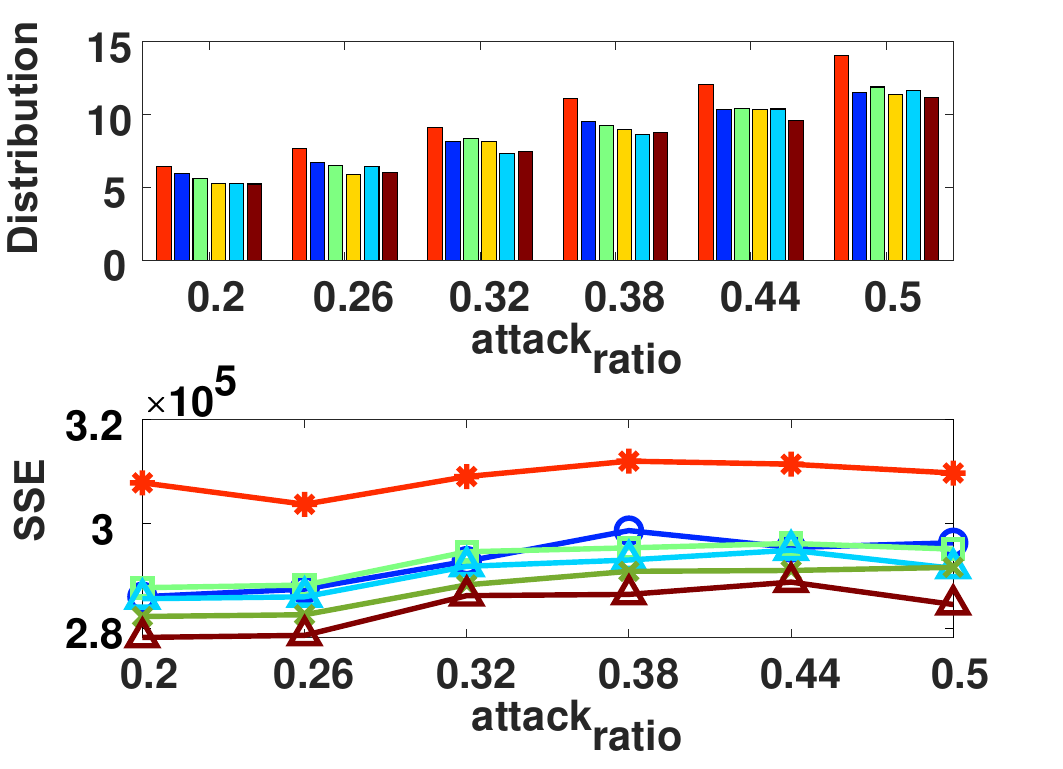}}\hspace{-0.02\textwidth}
    \subfigure[$\textbf{Vehicle}_{[0.2,0.5]}$]{\includegraphics[width=0.32\textwidth]{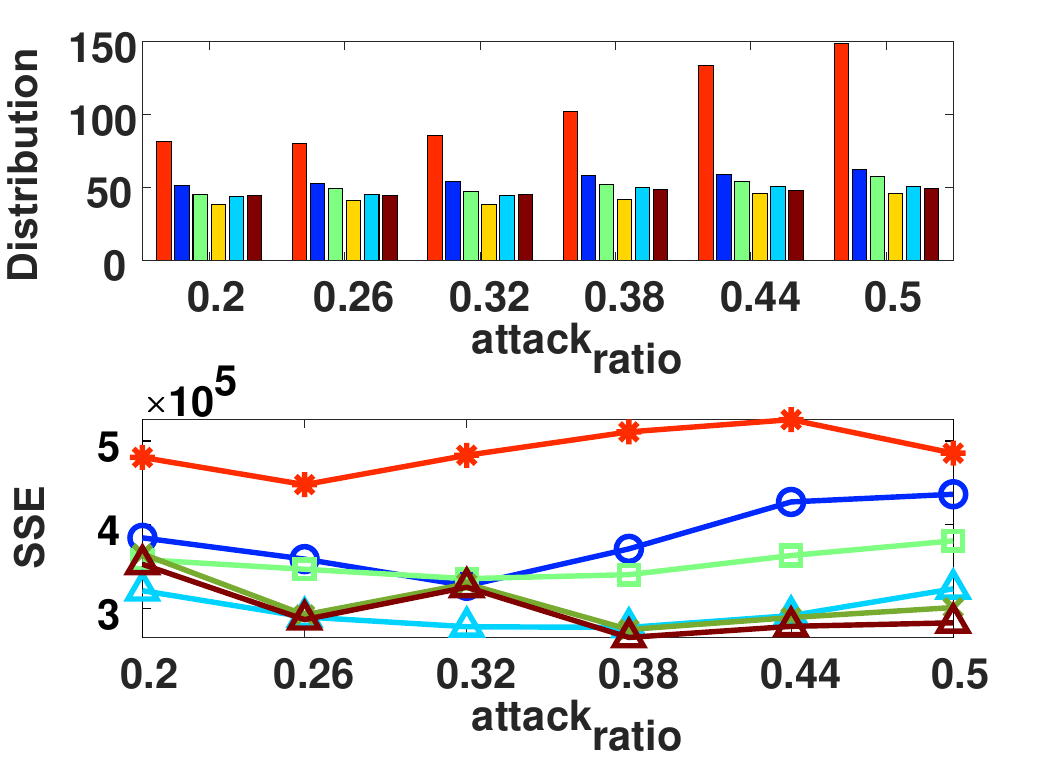}}\hspace{-0.02\textwidth}
    \subfigure[$\textbf{Letter}_{[0.2,0.5]}$]{\includegraphics[width=0.32\textwidth]{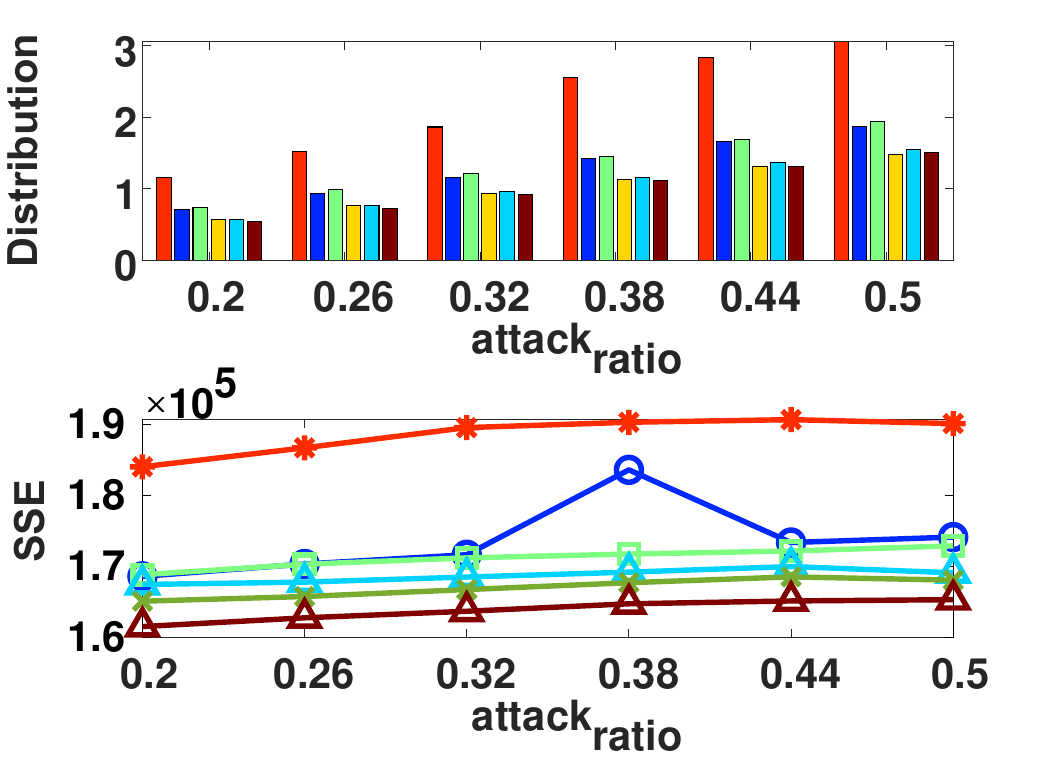}}\hspace{-0.02\textwidth}

    \caption{K-means clustering results over \textbf{Control}, \textbf{Vehicle}, and \textbf{Letter}, Tth=0.97}
    \label{fig:NE097}
    \vspace{-4mm}
\end{figure*}

\subsection{Stackelberg Equilibrium Results on SVM and SOM Classifier}This subsection validates the on labelled datasets with regard to Support Vector Machine (SVM) and Self-Organizing Map (SOM) classifiers, respectively. SVM and SOM are both classifiers included within MATLAB. We process various datasets and use them as inputs, directly showcasing the classification results. Specifically, we set the number of neurons in SOM to $20 \times 20 = 400$. The color depth between adjacent neurons represents their distance, with darker colors signifying greater distances between neurons. All elements classified into the same class have relatively small distances between them.

The SVM experiment is carried out exclusively on \textbf{Control} \textbf{(with labels)}. The parameters are fixed at Tth=0.95 and attack ratio=0.4. Fig. \ref{fig:groundtruth} (a) illustrates the ground truth of SVM classification, while Fig. \ref{fig:SVM} provides a comparison of SVM classification methods. The results are quite clear: the ground truth achieves an average accuracy of 96.8\%, while the various approaches under comparison yield respective accuracies of 95.5\%, 95.1\%, 94.9\%, 96.1\%, 95.6\%, and 95.7\%. The first three strategies comprise Ostrich and two Baselines. It is evident that $Baseline_{static}$ exhibits the poorest performance, even falling behind Ostrich. $Baseline_{0.9}$ also underperforms compared to Ostrich, a consequence of trimming excessive amounts of useful data. Our three approaches outperform others in terms of accuracy.

\begin{figure}
\centering
\subfigure[Ground truth of SVM]{\includegraphics[width=0.22\textwidth]{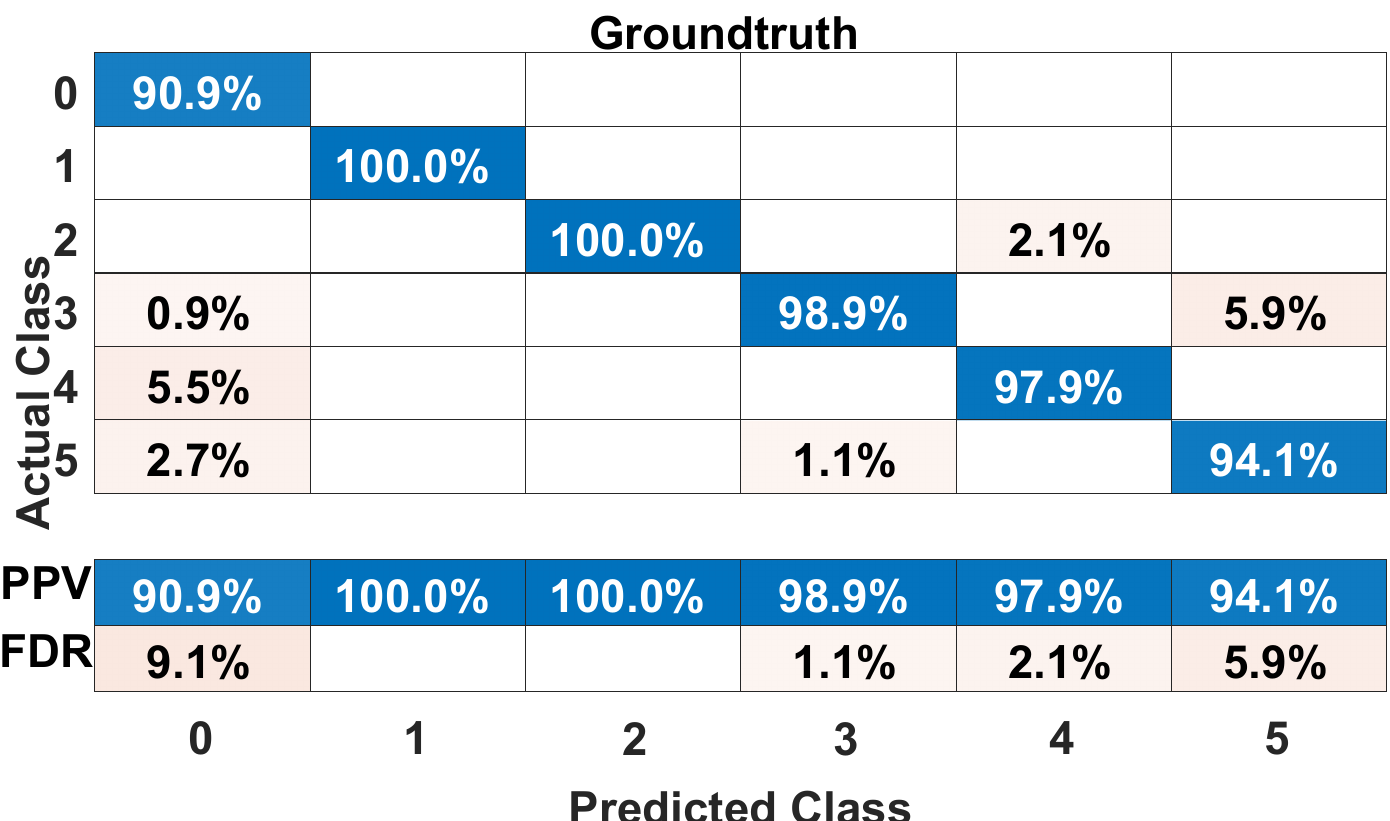}}
\subfigure[Ground truth of SOM]{\includegraphics[width=0.22\textwidth]{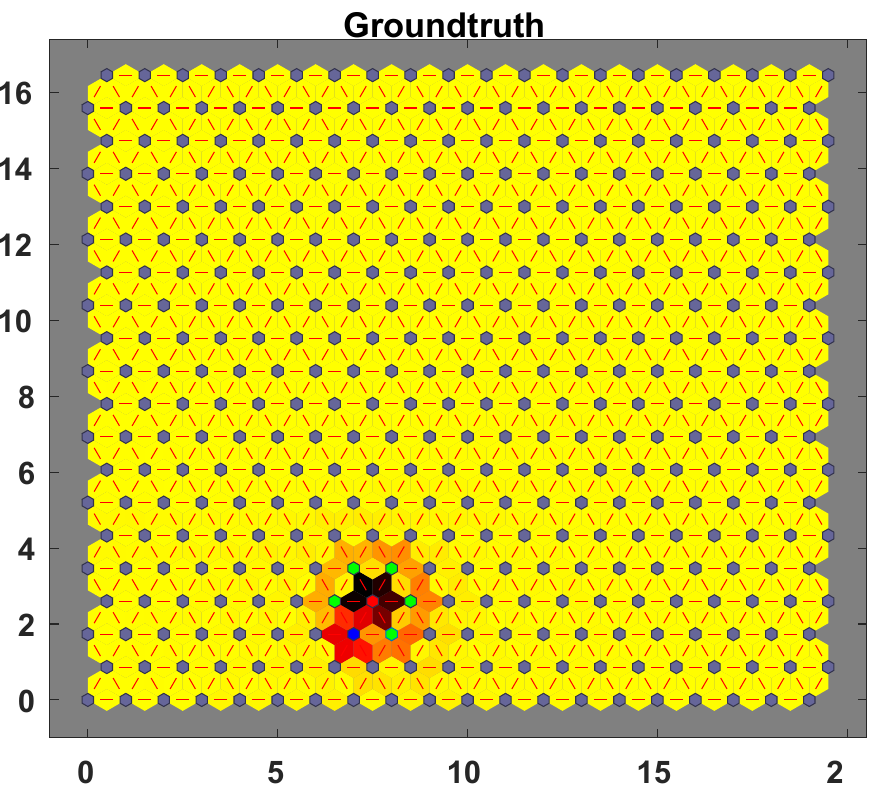}}
\caption{The ground truth of SVM and SOM classification}
\label{fig:groundtruth}
\vspace{-6mm}
\end{figure}

\begin{figure}[]
    \centering

    \subfigure[$Ostrich$]{\includegraphics[width=0.22\textwidth]{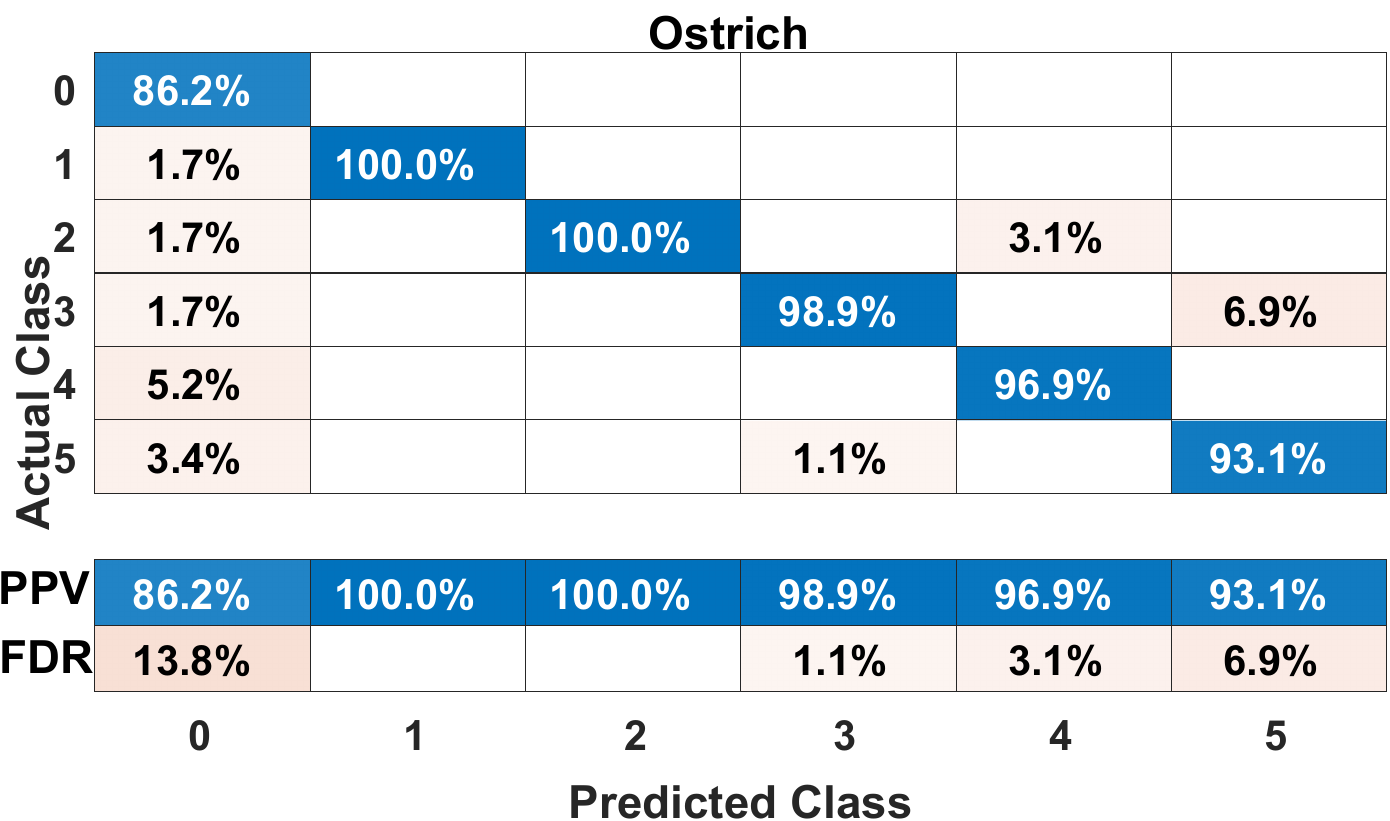}}\hspace{-0.00\textwidth}
    \subfigure[$Baseline_{0.9}$]{\includegraphics[width=0.22\textwidth]{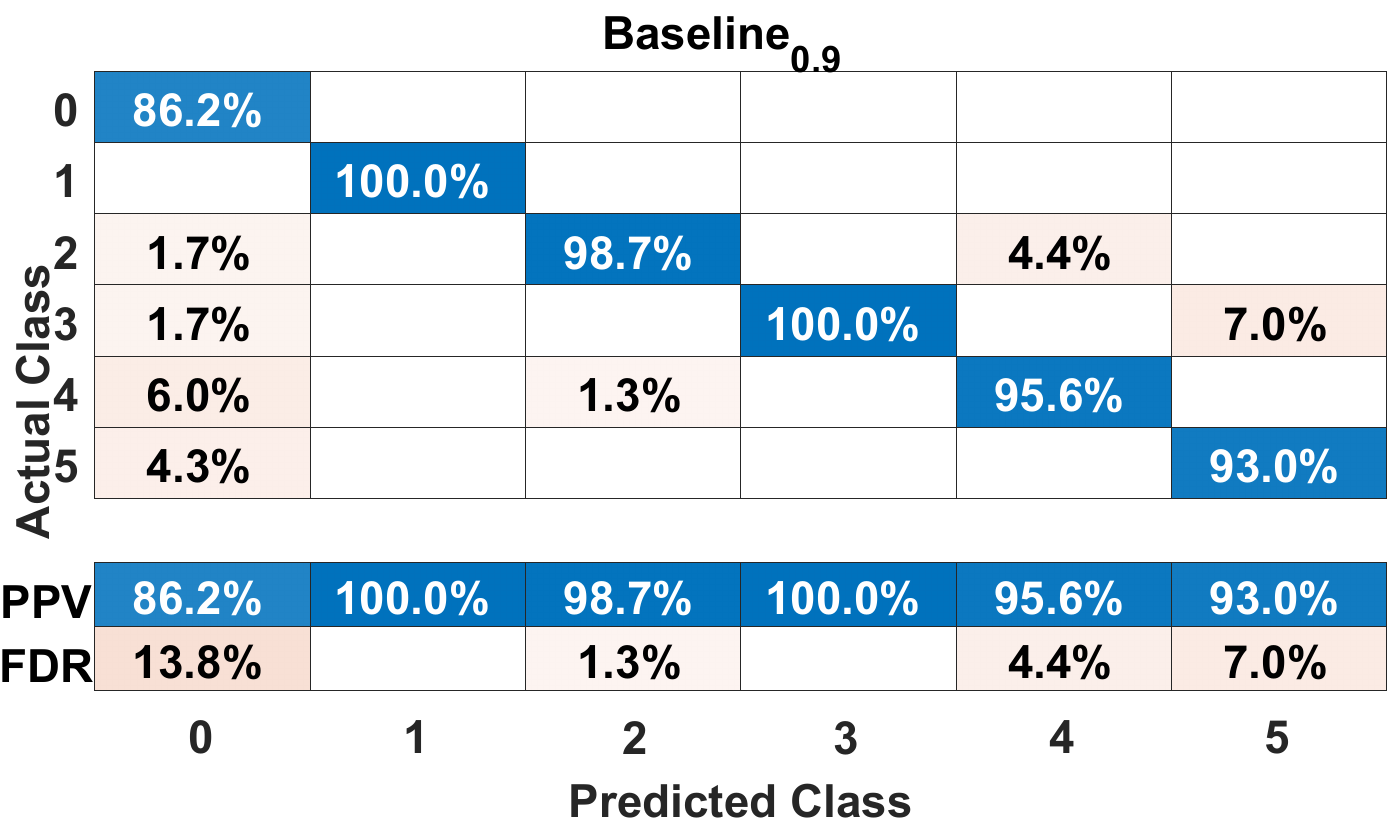}}\hspace{-0.01\textwidth}
    \subfigure[$Baseline_{static}$]{\includegraphics[width=0.22\textwidth]{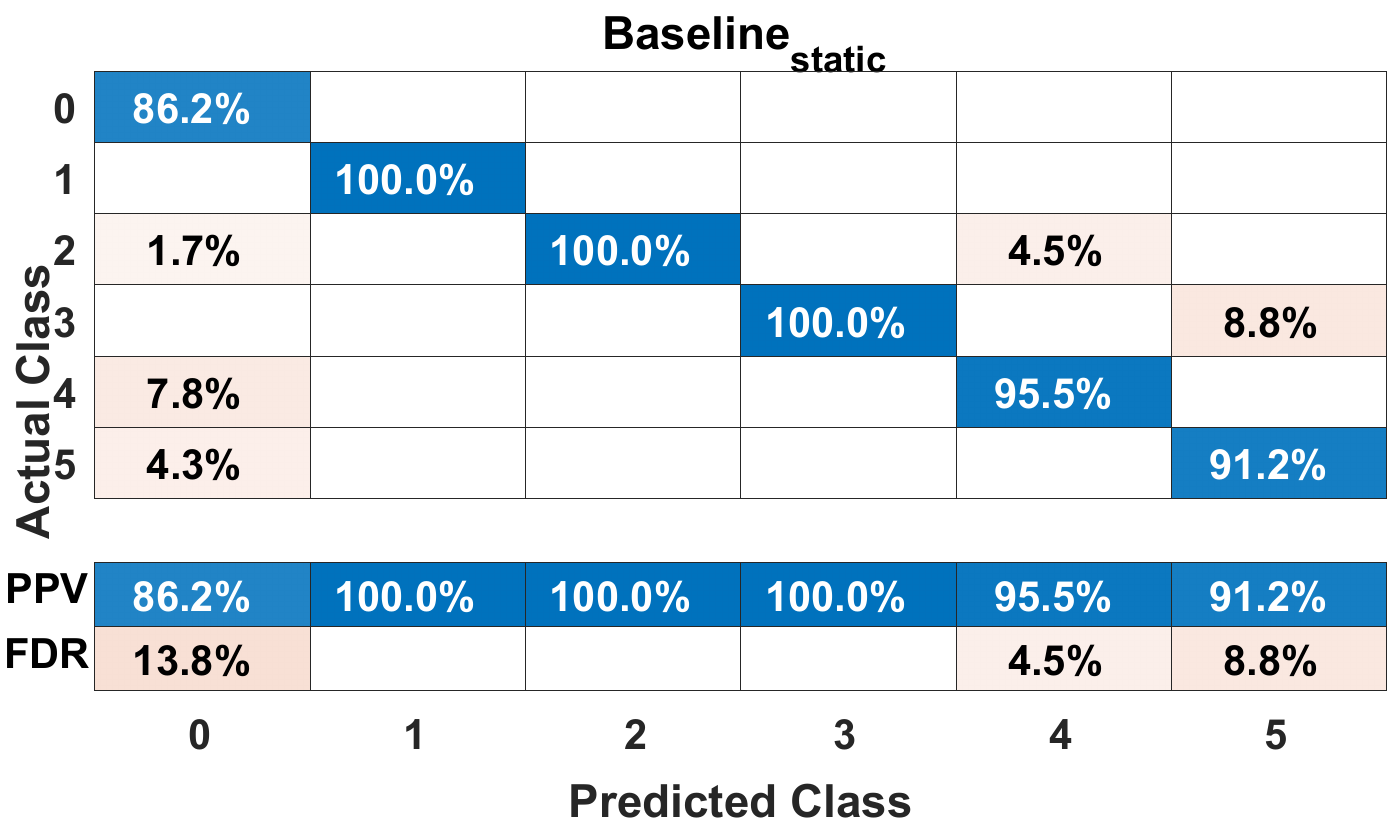}}\hspace{-0.00\textwidth}
    \subfigure[$Titfortat$]{\includegraphics[width=0.22\textwidth]{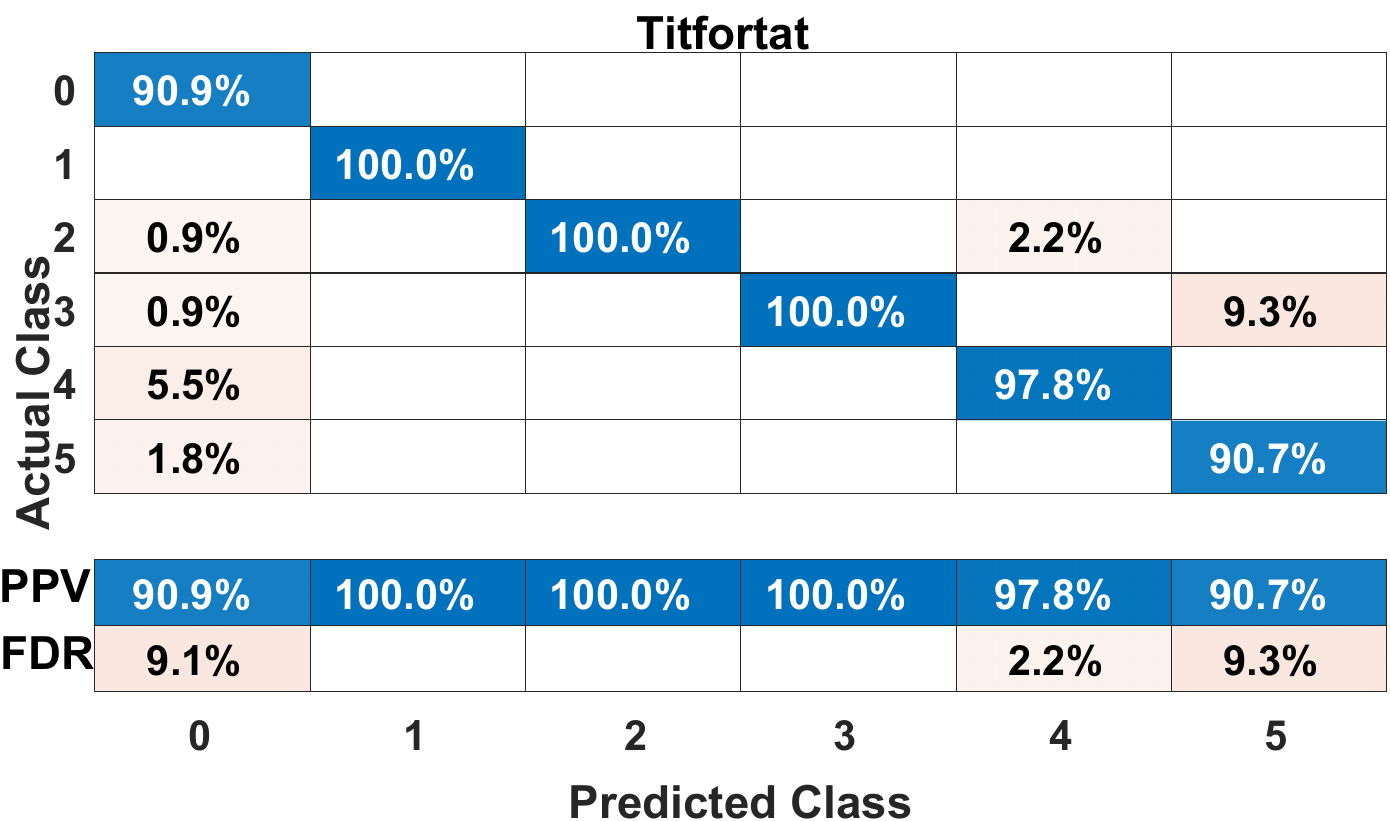}}\hspace{-0.01\textwidth}
    \subfigure[$Elastic_{0.1}$]{\includegraphics[width=0.22\textwidth]{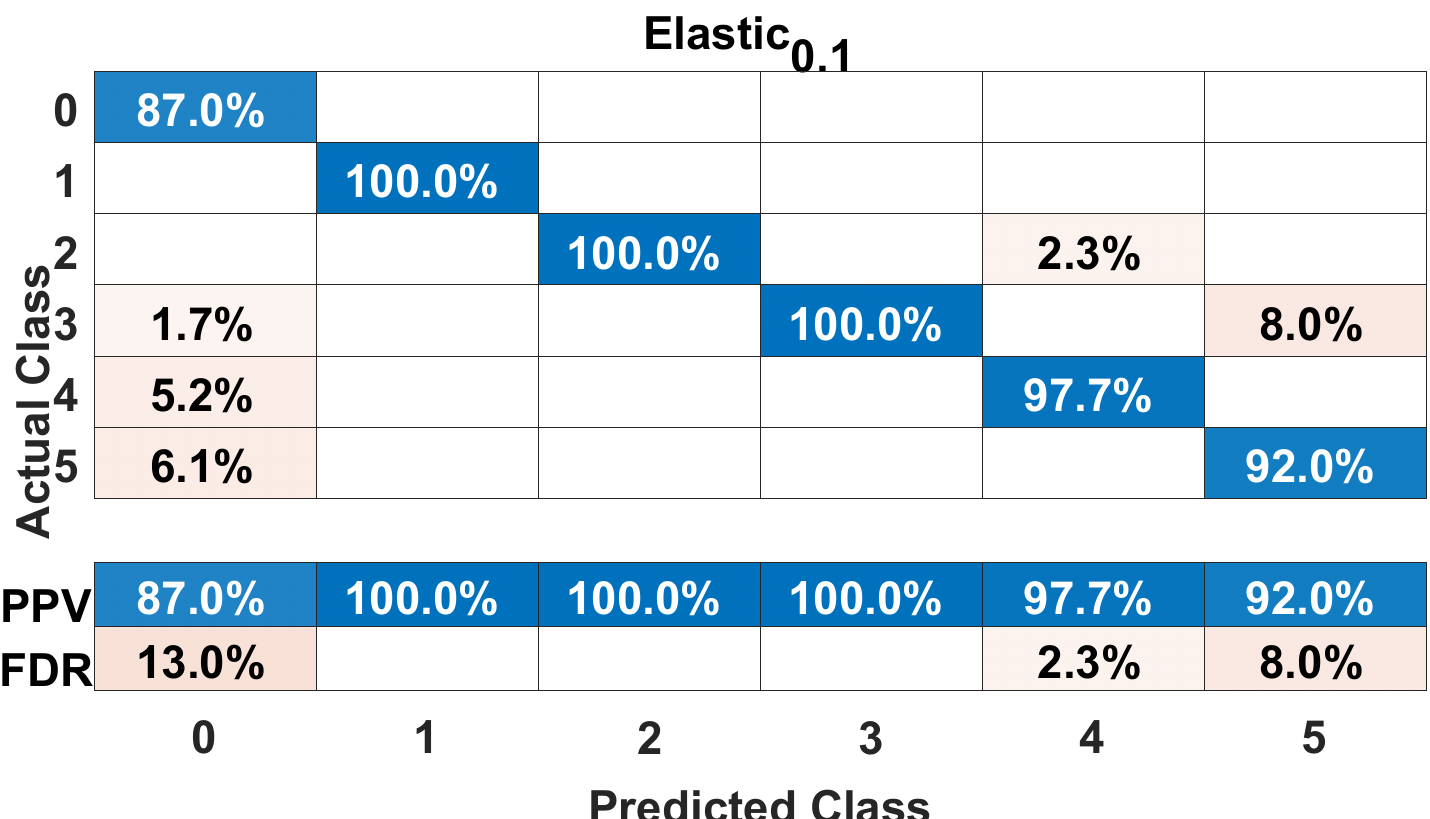}}\hspace{-0.00\textwidth}
    \subfigure[$Elastic_{0.5}$]{\includegraphics[width=0.22\textwidth]{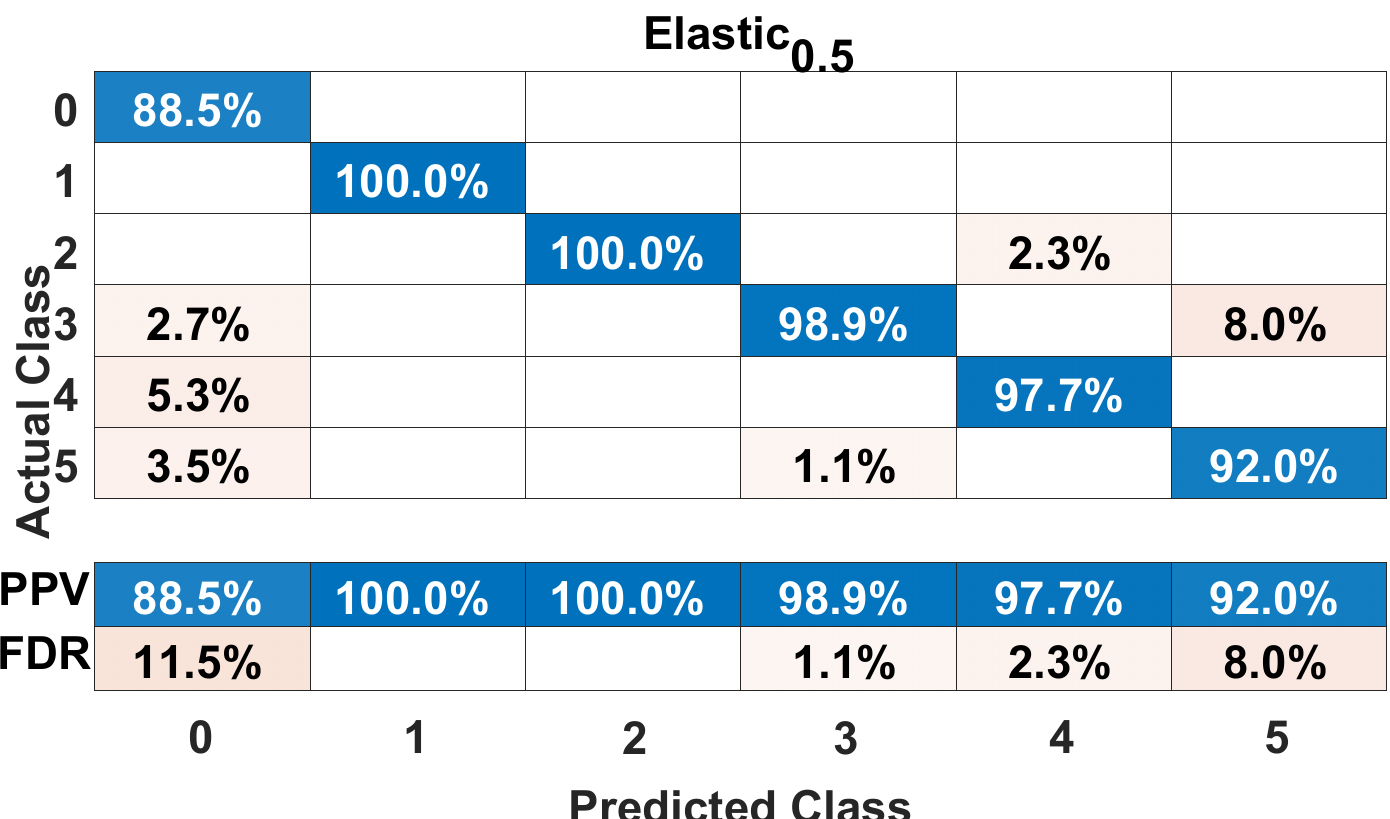}}
    \caption{Comparison of SVM classification}
    \label{fig:SVM}
    \vspace{-3mm}
\end{figure}

We carry out SOM classification on \textbf{Creditcard}, which contains credit card consumption data. The ground truth classification of this dataset, divided into four classes, is depicted in Fig. \ref{fig:groundtruth} (b). The classification results exhibit significant skewness and can be interpreted as follows:

The vast majority of data points belong to the same class, signifying the general public. The two isolated points, colored red and blue, are notably distant from other classes, representing fraudulent and premium users, respectively. The figure also includes five green points that symbolize a distinct category. These points are distant from both fraudulent and premium users, so they exhibit behaviors different from the general public. We can reasonably infer that these data points represent a segment of the general public with potential to evolve into high-value customers over time.

Fig. \ref{fig:SOM} presents a comparison of SOM classification results. We observe that Ostrich entirely disregards the large class corresponding to the green points. $Baseline_{0.9}$ performs worse than Ostrich, as it not only failed to differentiate the class corresponding to the green points but also lost the unique characteristics of the two isolated smaller classes. Though $Baseline_{static}$ successfully divides the data into four classes, it only includes a single isolated point, and the other three classes are overrepresented. Titfortat omits one isolated point and expands the area of the original green class. $Elastic_{0.1}$ and $Elastic_{0.5}$ each drop one isolated point but effectively represent the unique characteristics of the original class corresponding to the green points.

\begin{figure}[]
    \centering
    \subfigure[$Ostrich$]{\includegraphics[width=0.24\textwidth]{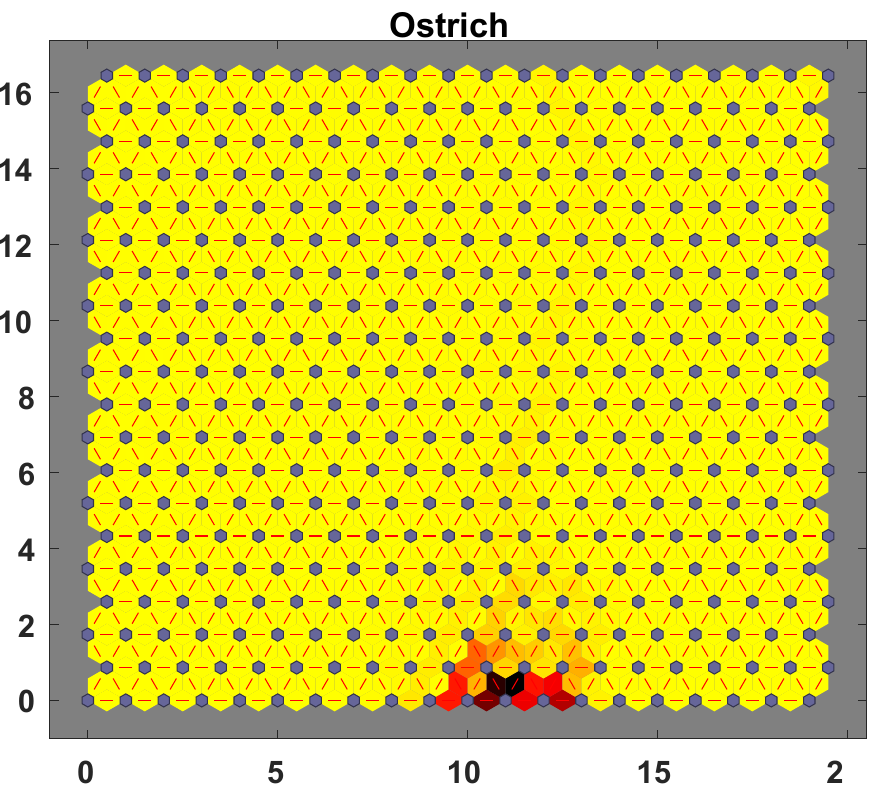}}\hspace{-0.01\textwidth}
    \subfigure[$Baseline_{0.9}$]{\includegraphics[width=0.24\textwidth]{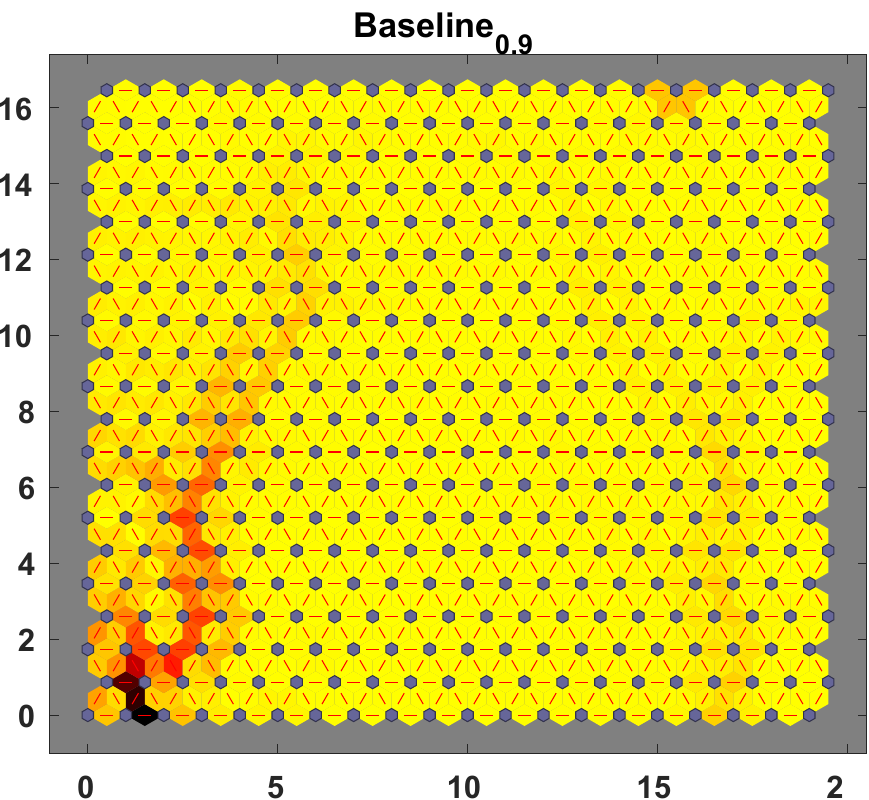}}\hspace{-0.01\textwidth}
    \subfigure[$Baseline_{static}$]{\includegraphics[width=0.24\textwidth]{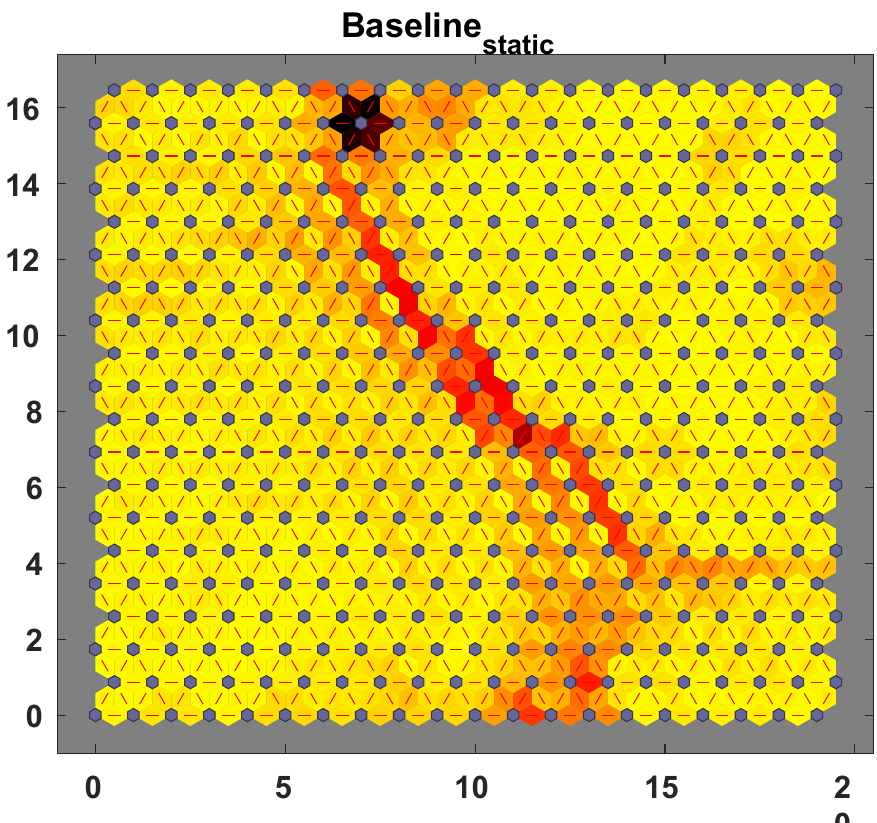}}\hspace{-0.01\textwidth}
    \subfigure[$Titfortat$]{\includegraphics[width=0.24\textwidth]{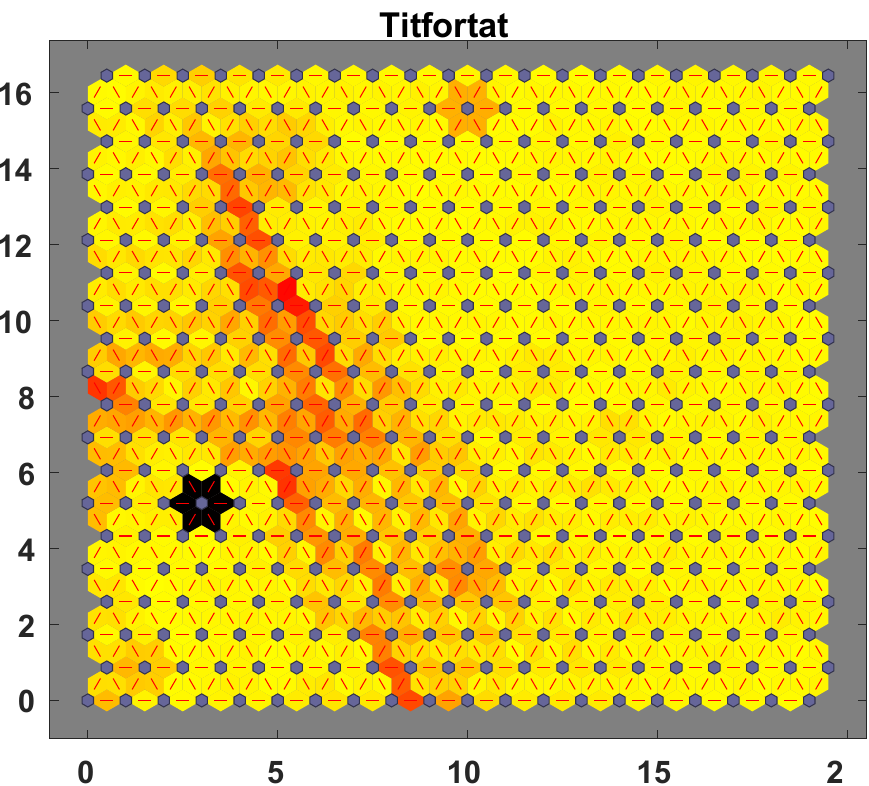}}\hspace{-0.01\textwidth}
    \subfigure[$Elastic_{0.1}$]{\includegraphics[width=0.24\textwidth]{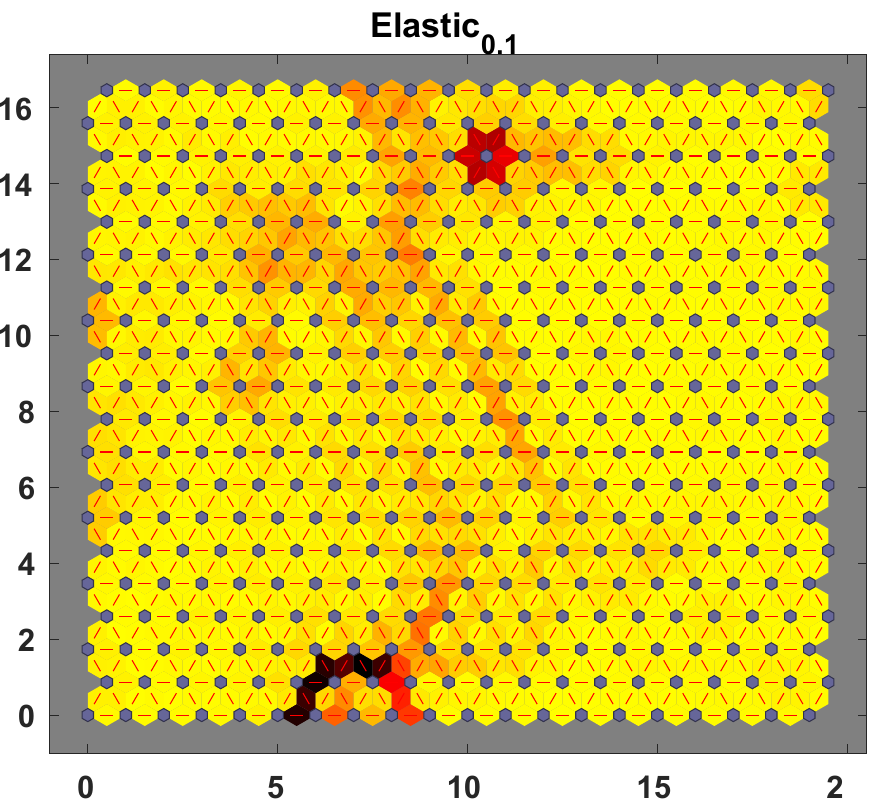}}\hspace{-0.01\textwidth}
    \subfigure[$Elastic_{0.5}$]{\includegraphics[width=0.24\textwidth]{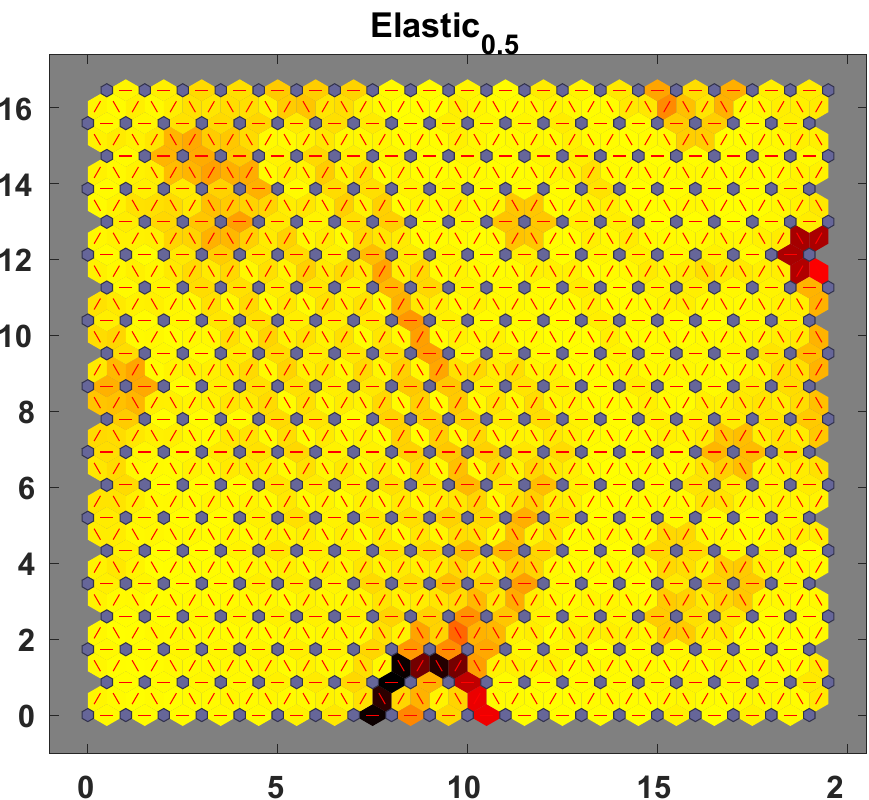}}
    \caption{Comparison of SOM classification}
    \label{fig:SOM}
    \vspace{-6mm}
\end{figure}

\subsection{Non Equilibrium Results and Cost Analysis}This subsection evaluates the utility under conditions where the adversary opts not to follow Stackelberg equilibrium strategies. The experiment is conducted on \textbf{Control} with attack ratio 0.2, involving 20 rounds of games. Given that all strategies can be represented as mixed strategies comprised of two linear combinations, we establish the 99th and 90th percentiles as the bases for these combinations, manipulated by parameter $p$. Poison values are injected at the 99th percentile with a probability of $p$ and at the 90th percentile with a probability of $1-p$.

In order to examine the early termination of Titfortat, we allow a redundancy of 5\%. This means that the stopping trigger condition is set to the initial observation where the ratio of poison values in a round exceeds $1-p+0.05$. Once this condition is triggered, the trimming position of Titfortat in subsequent rounds is permanently shifted to the 90th percentile. When $p=1$, it corresponds to an adversary who consistently adheres to the Stackelberg equilibrium strategy, while $p=0$ corresponds to an adversary who is both greedy and shortsighted. The effectiveness of any evasion strategy falls between these two extremes, controlled by parameter $p$. The experimental results are summarized in Table \ref{tab:NonNE}. The figures under the Titfortat and Elastic columns represent the proportion of untrimmed poison values in the remaining data, while the Average Termination Rounds denote the mean number of rounds the Titfortat strategy underwent before termination. The results suggest that an adversary following the Stackelberg equilibrium strategy realizes higher utility than one who does not.

We also conduct a cost analysis for the Elastic scheme, and the experimental results are presented in Table~\ref{tab:Elasticcost}. In the Elastic scheme, the handling of poison values involves imposing a penalty to the trimming threshold of the subsequent round based on a specified intensity and thus more rounds are required to achieve an equilibrium state. We define the cost of the Elastic scheme as the difference between the percentile of the data collector's soft trim and the actual percentile of the injected poison value before reaching equilibrium. The results displayed in the table represent the roundwise cost, which is the average cost over all rounds. As expected, the cost is higher in the initial rounds. However, as the Elastic strategy progressively adjusts the trimming threshold, the attacker's poison placement gradually approaches the equilibrium point, and the cost per round decreases accordingly. Hence, the roundwise cost diminishes with an increasing number of rounds, denoted by \(Round\_no\). Additionally, because the response intensity at \(k=0.5\) is greater than at \(k=0.1\), the former achieves equilibrium more rapidly, resulting in a lower roundwise cost.

\begin{table}
  \centering
  %\tiny
  %\scriptsize
  \caption{Non-equilibrium results and average termination rounds}
    \begin{tabular}{|c|c|c|c|}
    \hline
    $p$ & \multicolumn{1}{c|}{Average termination rounds} & \multicolumn{1}{c|}{Titfortat} & \multicolumn{1}{c|}{Elastic} \\
    \hline
    0     & 25    & 0.22727 & 0.22727 \\
    0.1   & 24.24 & 0.19157 & 0.22309 \\
    0.2   & 21.56 & 0.19645 & 0.21844 \\
    0.3   & 23.44 & 0.19264 & 0.21232 \\
    0.4   & 19.44 & 0.18381 & 0.20924 \\
    0.5   & 20.6  & 0.17904 & 0.20483 \\
    0.6   & 17.52 & 0.17363 & 0.19017 \\
    0.7   & 14.44 & 0.16874 & 0.17114 \\
    0.8   & 16.52 & 0.17011 & 0.15952 \\
    0.9   & 14.28 & 0.17041 & 0.15036 \\
    1     & 13    & 0.18182 & 0.14449 \\
    \hline
    \end{tabular}%
  \label{tab:NonNE}%
  \vspace{-6mm}
\end{table}%

\begin{table}[h]
\centering
\caption{Roundwise cost of $Elastic_{0.1}$ and $Elastic_{0.5}$}
\begin{tabular}{|c|c|c|}
\hline
Round\_no & k=0.5 (\%) & k=0.1 (\%) \\
\hline
5 & 0.608\% & 0.8\% \\
10 & 0.30404\% & 0.43281\% \\
15 & 0.20269\% & 0.28887\% \\
20 & 0.15202\% & 0.21667\% \\
25 & 0.12162\% & 0.17333\% \\
30 & 0.10135\% & 0.14444\% \\
35 & 0.086869\% & 0.12381\% \\
40 & 0.07601\% & 0.10833\% \\
45 & 0.067565\% & 0.096296\% \\
50 & 0.060808\% & 0.086667\% \\
\hline
\end{tabular}
\label{tab:Elasticcost}
\vspace{-6mm}
\end{table}

\subsection{Performance under LDP perturbations}This subsection evaluates the effectiveness of our proposed approach in privacy protection scenarios where data are perturbed by LDP techniques. For comparison, we use the Expectation-Maximization Filter (EMF)~\cite{du2023differential} as a baseline. It serves as a filtering mechanism designed to mitigate the effect of poison values under LDP data collection. The experiment is conducted on \textbf{Taxi}, the same dataset used in~\cite{du2023differential}. For the adversary of EMF, we employ the input manipulation attack~\cite{cheu2021manipulation}, which is identified as a potent evasion strategy against detection mechanisms within LDP-driven data collection scenarios.

\begin{figure*}[]
    \centering
    \subfigure[Attack ratio=0.05]{\includegraphics[width=0.32\textwidth]{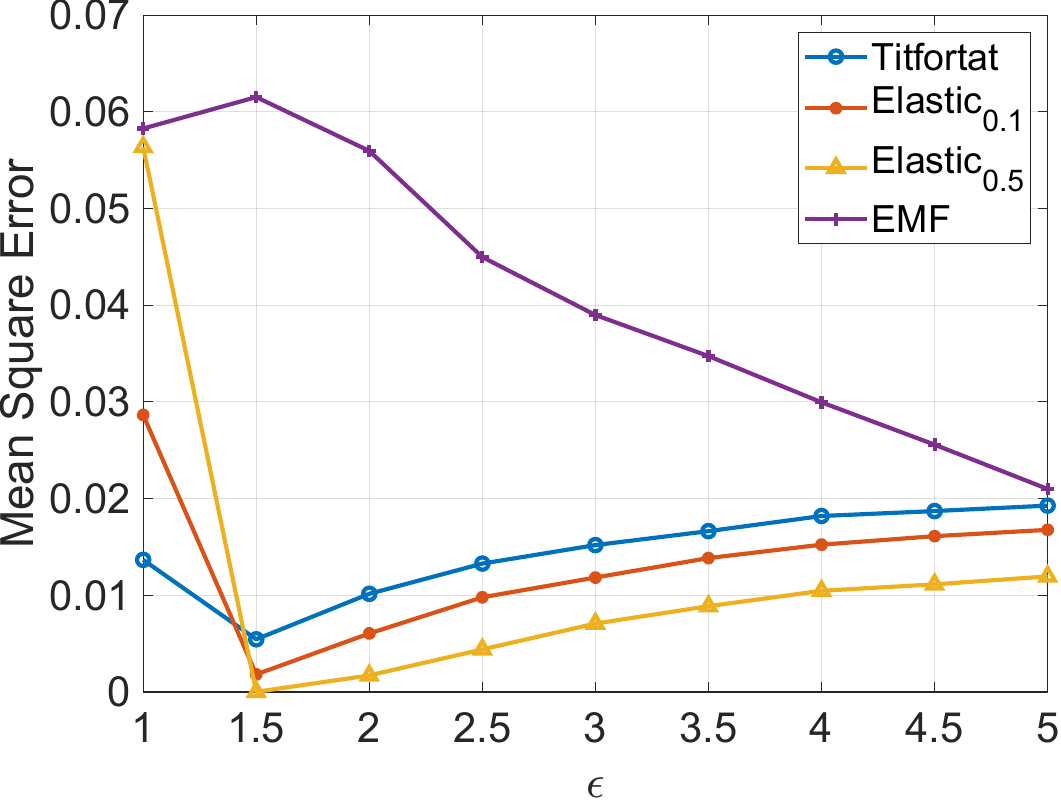}}%\hspace{-0.05\textwidth}
    \subfigure[Attack ratio=0.1]{\includegraphics[width=0.32\textwidth]{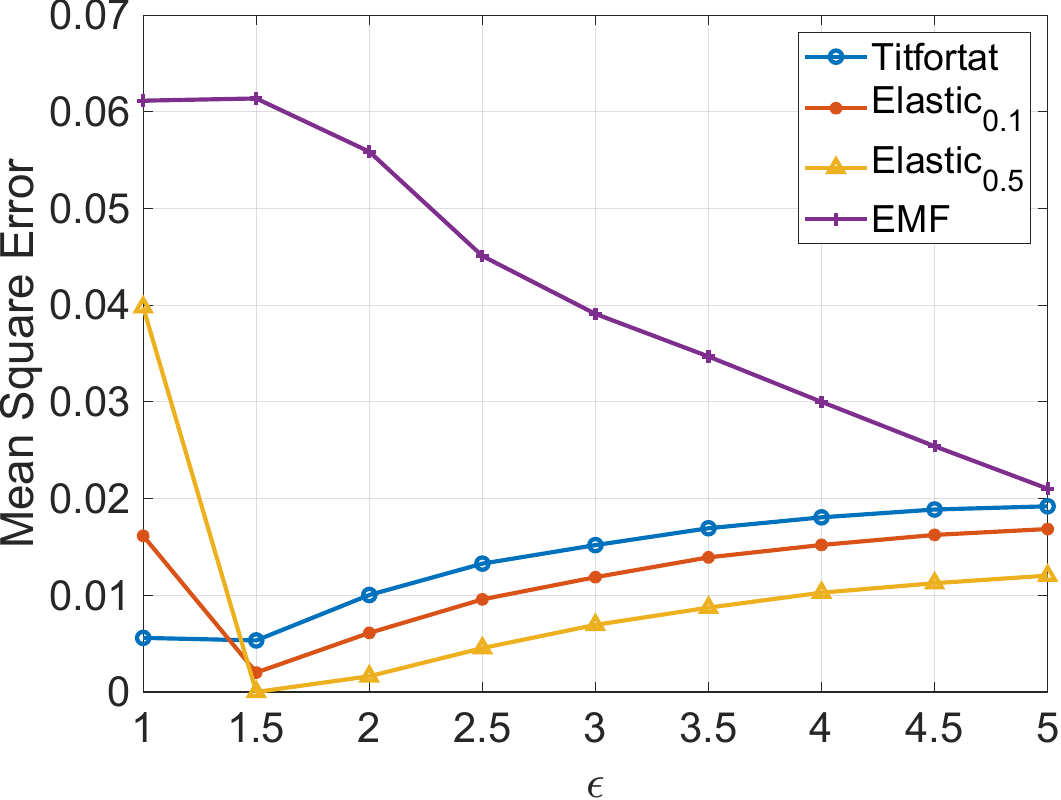}}%\hspace{-0.05\textwidth}
    \subfigure[Attack ratio=0.15]{\includegraphics[width=0.32\textwidth]{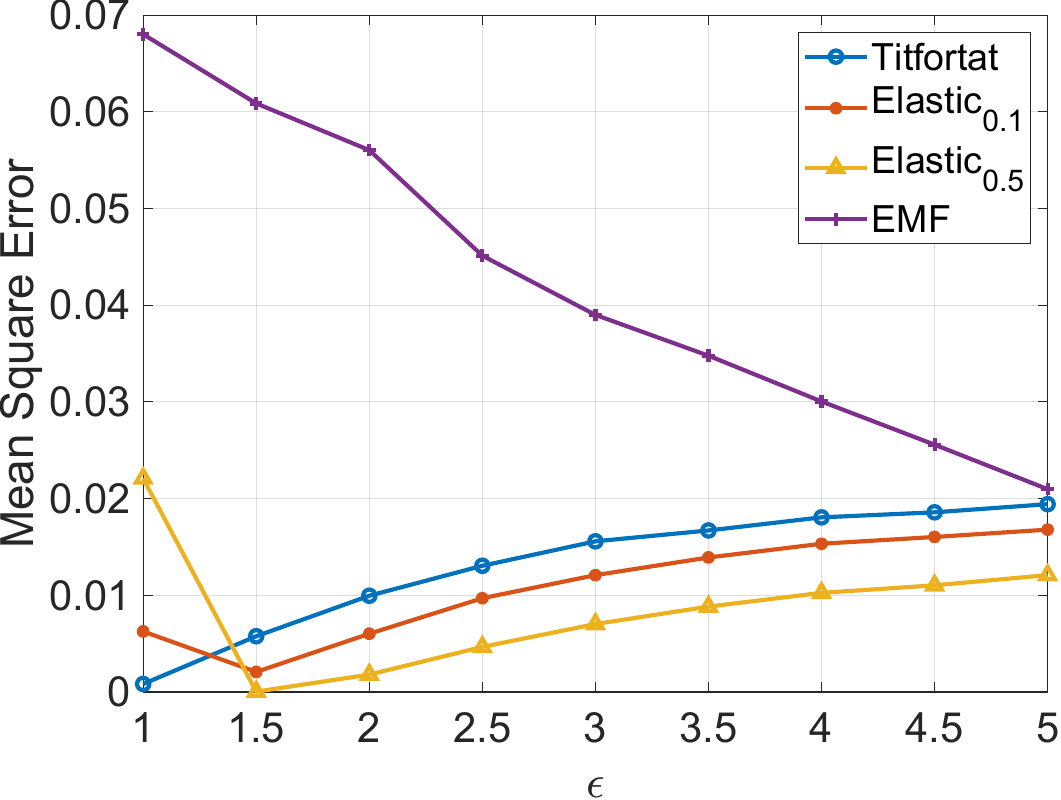}}%\hspace{-0.05\textwidth}

    \subfigure[Attack ratio=0.2]{\includegraphics[width=0.32\textwidth]{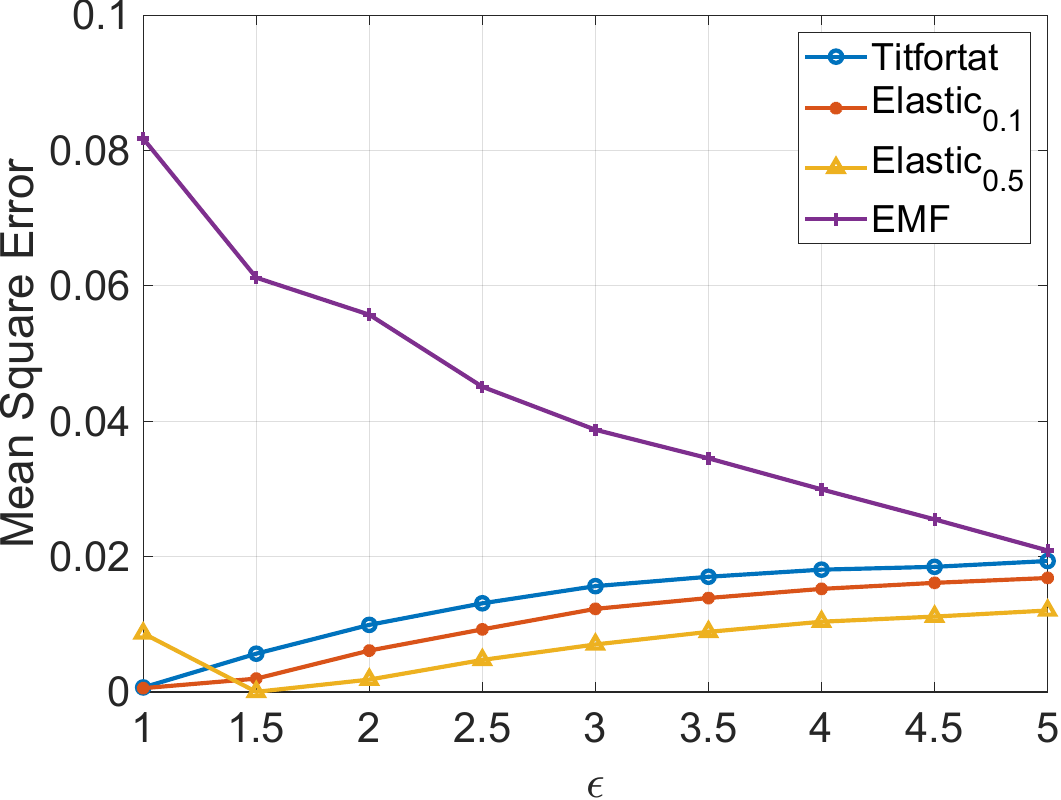}}%\hspace{-0.05\textwidth}
    \subfigure[Attack ratio=0.25]{\includegraphics[width=0.32\textwidth]{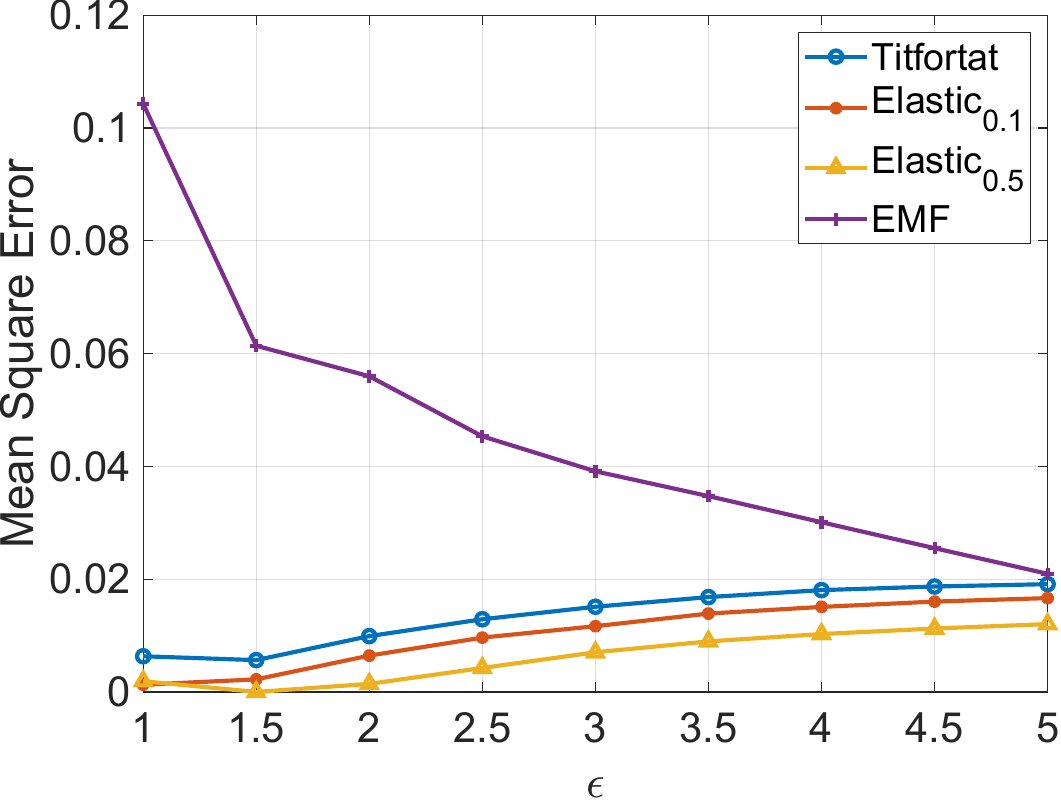}}%\hspace{-0.05\textwidth}
    \subfigure[Attack ratio=0.3]{\includegraphics[width=0.32\textwidth]{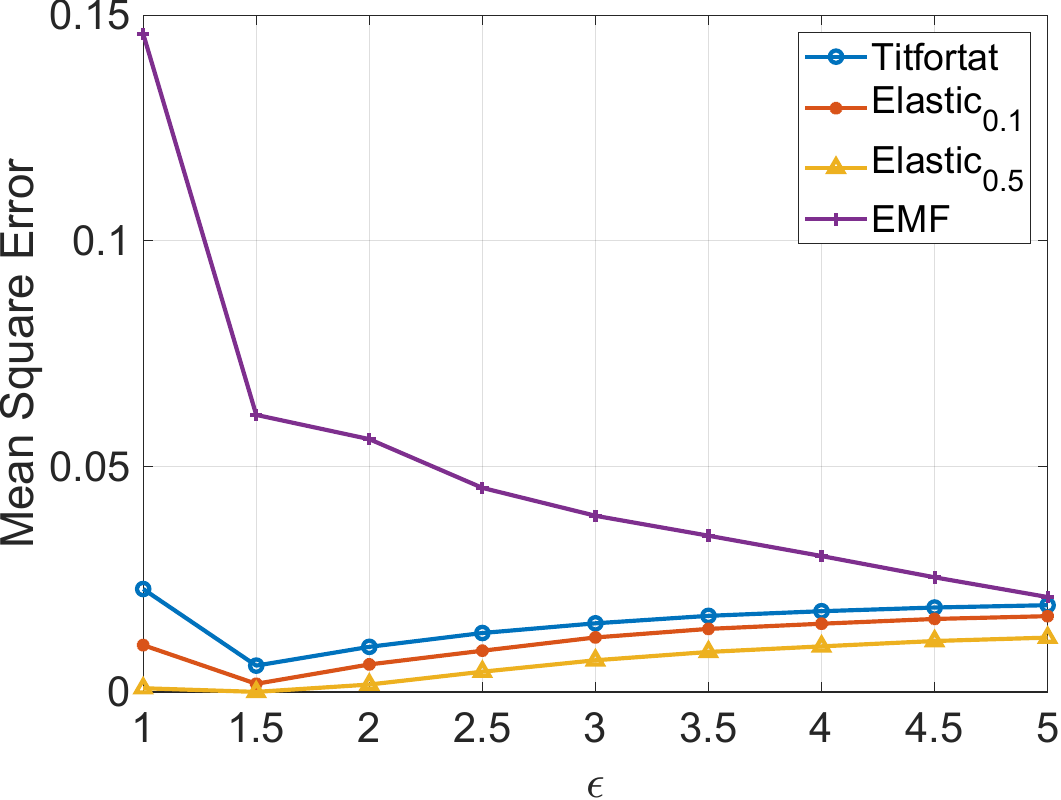}}%\hspace{-0.05\textwidth}

    \subfigure[Attack ratio=0.35]{\includegraphics[width=0.32\textwidth]{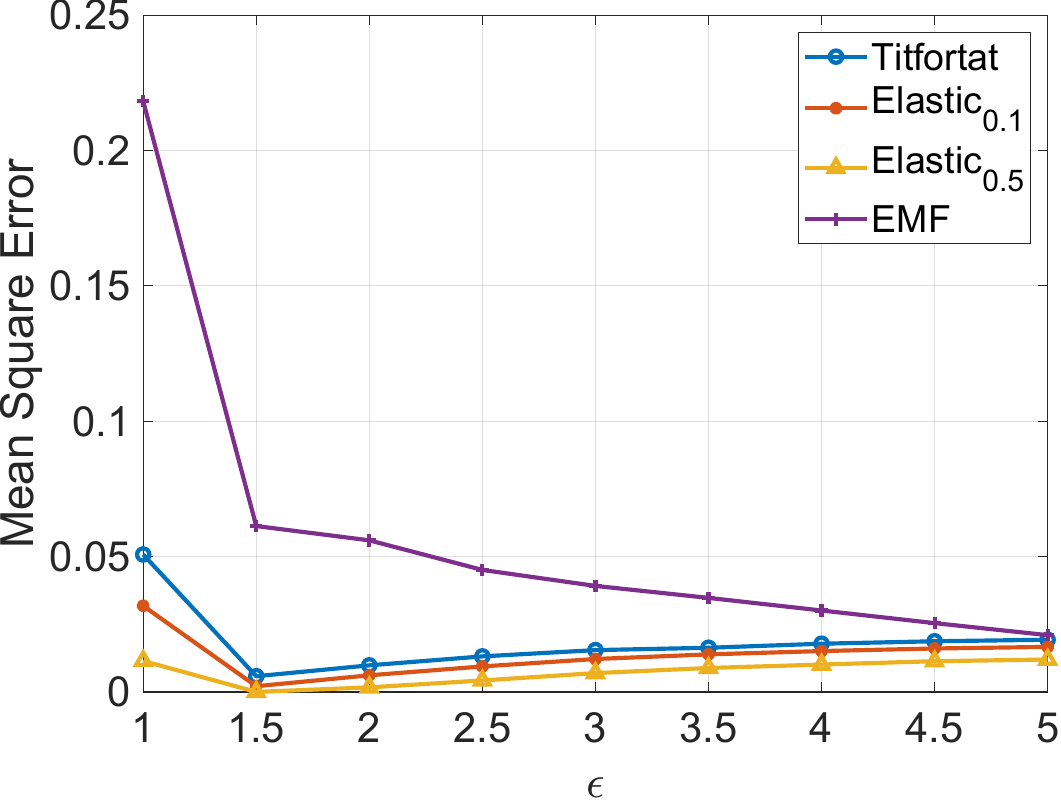}}%\hspace{-0.05\textwidth}
    \subfigure[Attack ratio=0.4]{\includegraphics[width=0.32\textwidth]{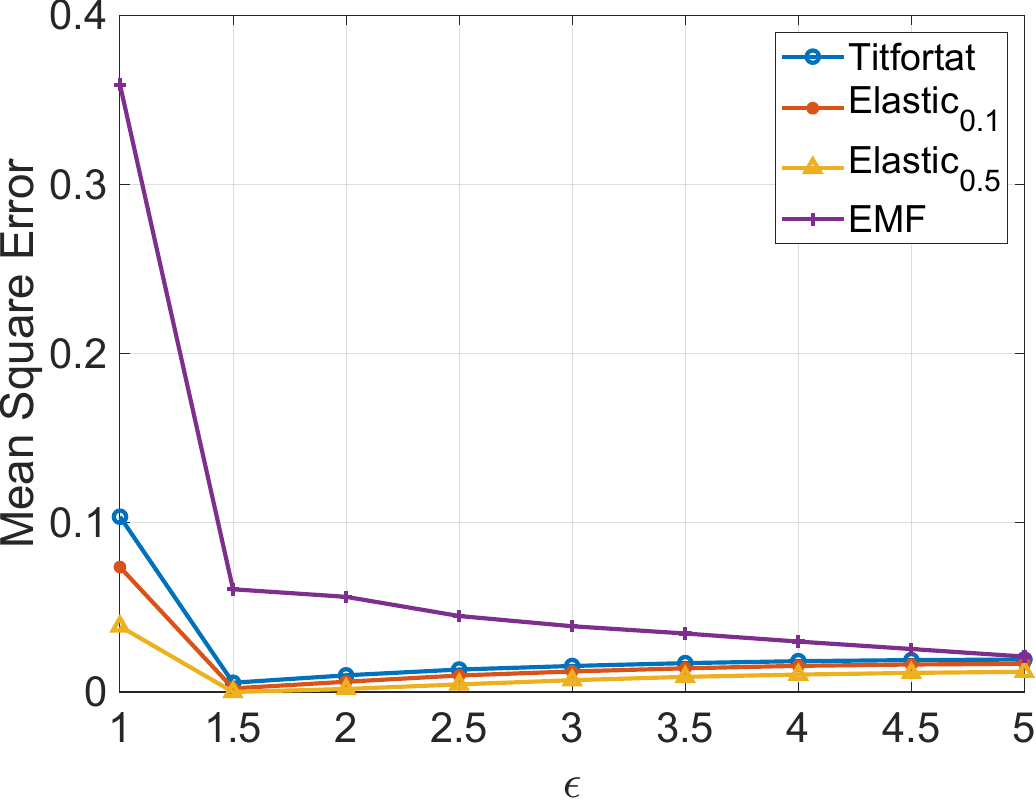}}%\hspace{-0.05\textwidth}
    \subfigure[Attack ratio=0.45]{\includegraphics[width=0.32\textwidth]{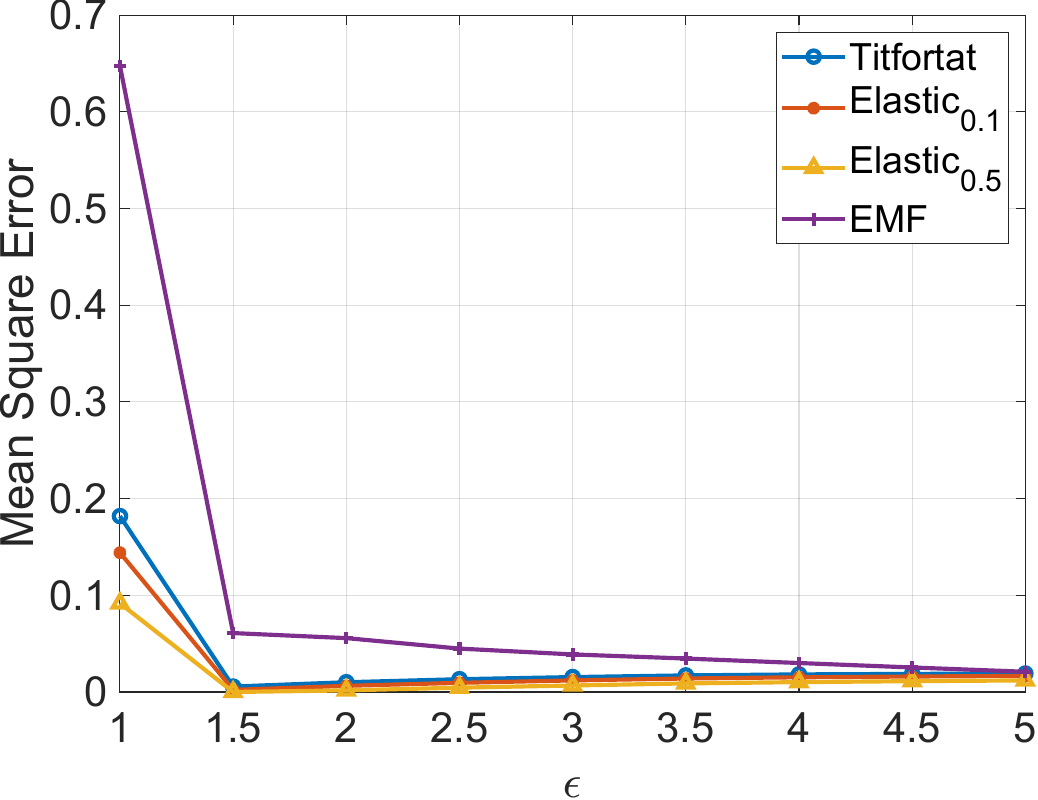}}%\hspace{-0.05\textwidth}

    \caption{Comparison of EMF and our proposed approaches}
    \label{fig:EMF}
    \vspace{-4mm}
\end{figure*}

Fig. \ref{fig:EMF} shows the results, where the x-axis represents the privacy budget $\epsilon$, and the y-axis indicates the Mean Square Error (MSE). As evidenced by the graph, in all parameter settings, the EMF consistently falls short of our scheme's performance. Notably, when $\epsilon$ is small, corresponding to a high perturbation intensity, the trimming scheme must accommodate increasing overhead due to false positives. This produces a notable inflection point around $\epsilon=1.5$ in the figure, an effect that becomes eminent when the attack ratio is small. 
\section{Related Work}
\label{Relatedwork}
In recent years, data poisoning attacks and their countermeasures have gained considerable attention, with numerous studies exploring various aspects, particularly in machine learning. Chen et al. \cite{chen2017zoo} devise a black-box attack method capable of bypassing defenses against data poisoning, emphasizing the need for more robust countermeasures. Biggio et al. \cite{biggio2012poisoning} examine poisoning attacks targeting SVM and introduced an effective attack strategy, highlighting the necessity for robust defenses in SVMs. Liu et al. \cite{liu2016delving} explore the transferability of adversarial examples, which are relevant to data poisoning attacks, and develop a black-box attack method based on this property. Steinhardt et al. \cite{steinhardt2017certified} introduce certification, offering guarantees on a model's robustness against data poisoning attacks, and present a certified training algorithm. Jagielski et al. \cite{jagielski2018manipulating} propose an optimal attack strategy for poisoning regression models and develop effective countermeasures to minimize the attacker's impact. Mei and Zhu \cite{mei2015using} employ a machine teaching approach to identify the most damaging data poisoning attacks on machine learning models, providing insights into both attack and defense strategies. Several other works also address data poisoning attacks and their countermeasures \cite{xiao2012adversarial, newell2014practicality, munoz2017towards}.

Manipulation attacks are more eminent in privacy-preserving scenarios (such as LDP \cite{cao2021data, huang2024ldpguard, sun2024ldprecover}) where the perturbations can amplify the effect of poison values, as the honest output follows a distribution, but the injected poison values may locate anywhere. Thus, for a single malicious user, the aggregated value will be larger than it should be, and it may even exceed the upper bound of the input domain. This implies that a small fraction of malicious users can obscure the distribution of honest user's inputs, and this can be extremely fatal when the privacy level is high, or the domain is large. A recent work shows how poor the performance of an LDP protocol can be under a malicious model \cite{cheu2021manipulation}. This work proposes a general manipulation attack in which Byzantine users can freely choose to report any poison values in the domain without following a distribution imposed by the LDP perturbation.

A special case of the general manipulation attack is the input manipulation attack, in which adversaries counterfeit some poison values before perturbation and strictly follow the LDP perturbation protocol. This can be treated as a special case of general manipulation attack with strong evasion, as it provides deniability for malicious users. If it is not possible to question individual users, the poison values are also indistinguishable from honest ones. While this evasion makes the poison values harder to detect, it also degrades the strength of the attack compared to general manipulations. This issue has received much attention, and some attempts have been made towards this problem in recent years. \cite{cao2021data} proposes a robust defense against poisoning attacks in federated learning systems by leveraging Byzantine-resilient aggregation methods. The authors focus on developing a method to detect and mitigate the impact of attackers who send poisoned model updates during the aggregation process. \cite{du2023differential} is an attempt to defend against general colluding attackers in LDP data collection. By exploiting the differences in behavior between attackers and normal users, a maximum likelihood estimation can be utilized to recover an attack distribution based on the collected data. However, this approach has a limitation, as it cannot address situations where attackers intentionally mimic the behavior of normal users. \cite{li2022fine} investigates the problem of data poisoning attacks in the context of graph neural networks. While they study the robustness of graph neural networks under poisoning attacks and propose a novel defense mechanism called robust graph convolutional networks, the evasive adversaries are seldom studied and modeled in a comprehensive framework. 
\section{Conclusion \& Future Work}
\label{Conclusion}
This paper presents a comprehensive game-theoretic model to counter online data poisoning attacks by establishing a viable Stackelberg equilibrium. We utilize the trimming strategy for defense, apply theoretical physics principles to construct an analytical model, and extend the adaptability in privacy-preserving systems with non-deterministic utility functions. Our experimental results, derived from various real-world datasets, validate the effectiveness of our approach. In future work, we plan to derive parameters corresponding to other game-theoretic approaches tailored for multi-round prisoner's dilemma scenarios, as well as the associated Elastic strategies. We also aim to construct theoretical frameworks for games with incomplete information pertinent to black-box models and provide corresponding experimental results.

\section*{Acknowledgement}
This work was supported by the National Natural Science Foundation of China (Grant No: 92270123, 62072390 and 62372122), and the Research Grants Council, Hong Kong SAR, China (Grant No:  15203120, 15226221, 15209922 and C2004-21GF).

\bibliographystyle{plain}
\bibliography{references}

\end{document}